\documentclass[twoside,11pt]{article}

\usepackage{jmlr2e}

\usepackage[english]{babel}
\usepackage{amsmath}
\usepackage{amstext}
\usepackage{amssymb}
\usepackage{mathtools}
\usepackage{algorithm}
\usepackage{algorithmic}
\usepackage{color}
\usepackage[dvipsnames]{xcolor}
\usepackage[normalem]{ulem}

\usepackage{multicol}
\usepackage{lipsum}
\usepackage{tcolorbox}
\tcbuselibrary{skins,breakable}
\usetikzlibrary{shadings,shadows}

\newcommand{\R}{\mathbb{R}}

\newcommand{\supp}{\mathop{\rm supp}}
\newcommand{\av}{\mathop{\rm av}}
\newcommand{\Cov}{\mathop{\rm Cov}}
\newcommand{\Var}{\mathop{\rm Var}}
\newcommand{\Hes}{\mathop{\rm Hes}}
\newcommand{\cspan}{\mathop{\rm span}}
\newcommand{\softmax}{\mathop{\rm softmax}}
\newcommand{\spath}{\mathop{\rm path}}
\newcommand{\opt}{\mathop{\rm opt}}
\newcommand{\Law}{\mathop{\rm Law}}

\newcommand{\Unif}{\mathop{\rm U}}

\newcommand{\T}{\mathop{\rm T}}
\newcommand{\ee}{\mathop{\rm e}}

\newcommand{\Id}{\mathop{\rm Id}}
\newcommand{\intr}{\mathop{\rm int}}

\newcommand{\argmin}{\mathop{\rm arg\, min}}

\newcommand{\bd}{{\boldsymbol d}}

\newcommand{\bepsilon}{{\boldsymbol \epsilon}}

\newcommand{\bxi}{{\boldsymbol \xi}}

\newcommand{\blambda}{{\boldsymbol \lambda}}
\newcommand{\btheta}{{\boldsymbol \theta}}

\newcommand{\ba}{{\bf a}}
\newcommand{\bx}{{\bf x}}

\newcommand{\inv}{{\mathsf{inv}}}

\newcommand{\sPr}{{\mathsf{Pr}}}

\newcommand{\sX}{{\mathsf X}}

\newcommand{\sA}{{\mathsf A}}

\newcommand{\sZ}{{\mathsf Z}}

\newcommand{\sE}{{\mathsf E}}

\newcommand{\D}{{\mathcal D}}

\newcommand{\F}{{\mathcal F}}

\newcommand{\C}{{\mathcal C}}

\newcommand{\Pnew}{{\mathcal P}}

\newcommand{\X}{{\mathcal X}}

\newcommand{\M}{{\mathcal M}}
\newcommand{\cR}{{\mathcal R}}

\newtheorem{assumption}{Assumption}

\newcommand{\ns}[1]{\color{black}{#1}\color{black}}
\newcommand{\tk}[1]{\textcolor{black}{#1}} 

\usepackage{lastpage}
\jmlrheading{NA}{NA}{1-\pageref{LastPage}}{NA; Revised NA
	}{NA}{24-0458}{Berkay Anahtarci, Can Deha Kariksiz, and Naci Saldi}
\ShortHeadings{Maximum Causal Entropy IRL in MFGs}{Anahtarci, Kariksiz and Saldi}

\firstpageno{1}

\allowdisplaybreaks

\begin{document}

\title{\ns{Maximum Causal Entropy IRL in Mean-Field Games and GNEP Framework for Forward RL}}

\author{\name Berkay Anahtarci \email berkay.anahtarci@ozyegin.edu.tr\\
       \addr Department of Natural and Mathematical Sciences\\
       \"{O}zye\u{g}in University\\
       \.{I}stanbul, Turkey
       \AND
       \name Can Deha Kariksiz \email deha.kariksiz@ozyegin.edu.tr\\     
       \addr Department of Natural and Mathematical Sciences\\
       \"{O}zye\u{g}in University\\
       \.{I}stanbul, Turkey
       \AND
       \name Naci Saldi \email naci.saldi@bilkent.edu.tr \\
       \addr Department of Mathematics\\
       Bilkent University\\
       Ankara, Turkey
       }
       
\editor{ }

\maketitle

\begin{abstract}
\tk{This paper explores the use of Maximum Causal Entropy Inverse Reinforcement Learning (IRL) within the context of discrete-time stationary Mean-Field Games (MFGs) characterized by finite state spaces and an infinite-horizon, discounted-reward setting. Although the resulting optimization problem is non-convex with respect to policies, we reformulate it as a convex optimization problem in terms of state-action occupation measures by leveraging the linear programming framework of Markov Decision Processes. Based on this convex reformulation, we introduce a gradient descent algorithm with a guaranteed convergence rate to efficiently compute the optimal solution. Moreover, we develop a new method that conceptualizes the MFG problem as a Generalized Nash Equilibrium Problem (GNEP), enabling effective computation of the mean-field equilibrium for forward reinforcement learning (RL) problems and marking an advancement in MFG solution techniques. We further illustrate the practical applicability of our GNEP approach by employing this algorithm to generate data for numerical MFG examples.}

\end{abstract}


\section{Introduction}
\label{intro}

Mean-field games (MFGs) were introduced to analyze continuous-time differential games with a vast number of agents. These agents strategically interact through a mean-field term, capturing the average distribution of the population's states. This concept was pioneered in the works of \citet{HuMaCa06} and \citet{LaLi07}.

In stationary MFGs, a typical agent characterizes the collective behavior of other agents \citep{WeBeRo05} through a time-invariant distribution, which leads to a Markov Decision Process (MDP) constrained by the state's stationary distribution. In this case, the equilibrium, referred to as the ``stationary mean-field equilibrium", involves a policy and a distribution satisfying the Nash Certainty Equivalence (NCE) principle \citep{HuMaCa06}. According to this principle, the policy should be optimal under a specified distribution, assumed to be the stationary infinite population limit of the mean-field term. Additionally, when the generic agent applies this policy, the resulting stationary distribution of the agent's state should align with this distribution. Under relatively mild assumptions, the existence of a stationary MFE can be proven using Kakutani's fixed-point theorem. Furthermore, it can be established that with a sufficiently large number of agents, the policy in a stationary MFE approximates a Nash equilibrium for a finite-agent scenario \citep{AdJoWe15}.

\tk{Classical MFG theory excels at computing equilibria when well-defined reward functions are provided, often employing forward reinforcement learning (RL) techniques \citep[see][]{LaPePeGiMuElGePi24}. However, specifying the reward function in MFG problems can be difficult in practical applications. Inverse Reinforcement Learning (IRL) addresses this issue by inferring the reward function from expert demonstrations, aiming to develop policies that allow agents to effectively imitate the expert behavior. This method
has significant real-world applications in diverse fields such as finance \citep{BeViTrRe23}, autonomous vehicles \citep{ChLiDi23}, and epidemic control \citep{DoGaGa22}.
}


The potential benefits of IRL are multifaceted. Firstly, by learning from demonstrations, IRL facilitates the incorporation of human knowledge and guidance into the learning process. This can be particularly valuable in situations where reward functions are complex or difficult to define. Secondly, IRL promotes interpretability and transparency in agent behavior. By inferring the reward function from demonstrations, we gain insights into the agent's motivations and decision-making process, which is useful in safety-critical applications where human-agent collaboration is essential. Finally, IRL has the potential to improve agent generalization to unseen scenarios. By focusing on the underlying objectives, IRL might lead to agents that can adapt and perform well in situations not explicitly experienced during training \citep[see][]{AdCoBe22}.

Several recent papers have tackled the IRL problem in the context of MFGs. \citet{YaYe18proc} reduce a specific MFG to an MDP and employs the principle of maximum entropy to solve the corresponding IRL problem. However, this reduction from MFG to MDP is only applicable in fully cooperative settings, where all agents share the same societal reward. In contrast, typical MFGs involve decentralized information structures and misaligned objectives between agents. To address this, \citet{YaLiLiHu22} formulate the IRL problem for MFGs in a decentralized and non-cooperative setting, then tackle it using a maximum margin approach. \tk{ Compared to these works, the most relevant one to ours is \citet{ChZhLiWi23}, making it important to highlight the differences. A key difference lies in the problem setting: \citet{ChZhLiWi23} focus on a finite-horizon cost framework, whereas our work addresses the infinite-horizon case. As we discuss in the appendix, this distinction is crucial because, in the infinite-horizon setting, the problem cannot be formulated over the set of path probabilities, which makes a maximum likelihood approach infeasible. Another fundamental difference concerns the assumptions on the expert’s trajectory. \citet{ChZhLiWi23} assume that the expert follows an entropy-regularized mean-field equilibrium. This assumption enables them to recover the expert’s policy using a likelihood-based optimization framework. In contrast, our approach assumes that the expert trajectory is drawn from a standard (non-regularized) mean-field equilibrium. Given this trajectory, we construct the mean-field term and the expected feature vector, then seek policies that explain the observed behavior while maximizing entropy. Hence, the solution to our optimization problem--together with the mean-field term--yields a standard mean-field equilibrium rather than an entropy-regularized one. A further distinction lies in the uniqueness of the recovered solution. In \citet{ChZhLiWi23}, the authors ensure the uniqueness of the solution upfront by leveraging the variational formula. This guarantees that their approach recovers a unique policy that explains the observed behavior. In contrast, uniqueness is not required in our framework, multiple policies may be consistent with the observed behavior, and our goal is to select the one with the highest entropy among them. Overall, while our approaches differ in their technical formulations and underlying assumptions, they are complementary in nature. Exploring deeper connections between these formulations presents an interesting avenue for future research.}

Closely related to the IRL framework is imitation learning (IL). It offers a complementary approach to the IRL framework. While IRL infers the reward function from demonstrations, IL directly learns a policy that mimics the desired behavior. This can be particularly beneficial when expert demonstrations are readily available and the reward function is difficult to define explicitly. Recent work by \citet{GiPaOlNiMaMa23} showcases the potential of IL applications in the context of MFGs.

It is worth mentioning that the aforementioned papers focus solely on finite-horizon cost structures, leading to convex optimization problems that employ classical maximum entropy principle and maximum margin approach. However, the classical maximum entropy principle used in \citet{ChZhLiWi23} generally does not work for infinite-horizon problems due to the fact that the distribution induced by the state-action process on the path space becomes ill-defined in such cases. To address this limitation, \citet{ZhBlBa18} introduce the maximum causal entropy principle for infinite-horizon problems within MDPs. \tk{In this paper, we address the maximum entropy IRL problem in the context of infinite-horizon MFGs, where the expert trajectory is assumed to arise from a standard (non-regularized) MFE. Unlike the aforementioned approaches, our framework applies the maximum entropy principle without relying on entropy regularization. Specifically, we estimate the mean-field term and expected feature vector from this trajectory, identifying policies that maximize entropy to explain the observed behavior. Unlike approaches that enforce uniqueness of the MFE, we select the highest entropy solution from among potentially multiple MFEs that align with the behavior.
}


\subsection{Contributions}
\tk{
\begin{itemize}
\item[1.] We introduce the maximum causal entropy IRL problem for discrete-time stationary MFGs. This is an extension of the framework initially established for MDPs in \citet{ZhBlBa18}. We address scenarios in both MDPs and MFGs where the optimality criterion is characterized by an unknown infinite-horizon discounted reward.
\item[2.] The Maximum Causal Entropy IRL problem for MFGs is inherently non-convex. To address this, we reformulate it as a convex optimization problem using a linear programming framework based on state-action occupation measures. Within this dual context, we demonstrate the smoothness and strong convexity of the objective function. Then, we apply a gradient descent algorithm with a constant learning rate to find the optimal solution to the dual problem, ensuring a guaranteed convergence rate (Section~\ref{me-irl-mfg}).
\item[3.] We propose an alternative approach that operates under more relaxed assumptions to overcome restrictive conditions typically required for contraction-based methods in computing MFE. By reframing the MFG problem as a Generalized Nash Equilibrium Problem (GNEP) through the linear programming (LP) framework of MDPs, we adapt an algorithm from the GNEP literature to develop a novel MFE computation algorithm. This approach reduces dependency on stringent conditions such as Lipschitz continuity and strong convexity. Hence, the GNEP formulation provides a more flexible framework that can handle a broader class of MFGs. We demonstrate the method's practical applicability by using this algorithm to generate data for numerical examples (Section~\ref{mfg-gnep}). \end{itemize}
}

\noindent\textbf{Notation.} 
For a finite set $\sE$, we let $\Pnew(\sE)$ denote the set of all probability distributions on $\sE$ endowed with the $l_2$-norm $\|\cdot\|$. For any $e \in \sE$, $\delta_e$ is the Dirac delta measure. For any $a,b \in \R^d$, $\langle a,b \rangle$ denotes the inner product. The notation $v\sim \nu$ means that the random element $v$ has distribution $\nu$.

\section{Maximum Causal Entropy Principle in MFGs}\label{me-irl-mfg}


\tk{
Drawing on the MDP fundamentals established in the appendix, this section extends the maximum causal entropy IRL framework to the realm of MFGs. We adapt the existing formulation for infinite-horizon MDPs to accommodate the unique characteristics of MFGs, specifically the interplay between individual agent behavior and the aggregate behavior of the entire population. This adaptation ensures that the principles of maximum causal entropy remain both valid and insightful within the MFG context.
}

A discrete-time stationary MFG is specified by
$
\left( \sX,\sA,p,r,\beta \right),
$
where $\sX$ is the finite state space and $\sA$ is the finite action space. The components $p : \sX \times \sA \times \Pnew(\sX) \to \Pnew(\sX)$ and $r: \sX \times \sA \times \Pnew(\sX) \rightarrow [0,\infty)$ are the transition probability and the one-stage reward function, respectively. Therefore, given current state $x(t)$, action $a(t)$, and state-measure $\mu$, the reward $r(x(t),a(t),\mu)$ is received immediately, and the next state $x(t+1)$ evolves to a new state probabilistically according to the following distribution:
$$
x(t+1) \sim p(\cdot|x(t),a(t),\mu).
$$

In MFGs, a state-measure $\mu \in \Pnew(\sX)$ represents the collective behaviour of the other agents, that is, $\mu$ can be considered as the infinite population limit of the empirical distribution of the states of other agents.\footnote{\tk{In classical mean-field game literature, the exogenous behavior of other agents is typically modeled by a time-varying state measure-flow $\{ \mu(t) \}$, where $\mu(t) \in \Pnew(\sX)$ for all $t$. This implies a non-stationary total population behavior. However, our work only focuses on the stationary case, assuming a constant state measure-flow $\mu(t) = \mu$ for all $t$.}} 

To complete the description of the model dynamics, we should also specify how the agent selects its action. To that end, a policy $\pi$ is a conditional distribution on $\sA$ given $\sX$, that is, $\pi: \sX \rightarrow \Pnew(\sA)$. Let $\Pi$ denote the set of all policies. 

Before presenting the optimality notion adapted to MFGs, we first introduce the discounted reward of any policy given any state measure.  
In discounted MFGs, for a fixed $\mu$, the infinite-horizon discounted-reward function of any policy $\pi$ is given by
\begin{align}
J_{\mu}(\pi,\mu_0) &= E^{\pi}\biggl[ \sum_{t=0}^{\infty} \beta^t r(x(t),a(t),\mu) \biggr], \nonumber 
\end{align}
where $\beta \in (0,1)$ is the discount factor and $x(0) \sim \mu_0$ is the initial state distribution. In this model, we define the set-valued mapping $\Psi : \Pnew(\sX) \rightarrow 2^{\Pi}$  as follows (here, $2^{\Pi}$ is the collection of all subsets of $\Pi$):
$$\Psi(\mu) = \{\hat{\pi} \in \Pi: J_{\mu}(\hat{\pi},\mu) = \sup_{\pi} J_{\mu}(\pi,\mu) \}.$$
The set $\Psi(\mu)$ indeed represents the optimal policies for a specified $\mu$.
Similarly, we define the set-valued mapping $\Lambda : \Pi \to 2^{\Pnew(\sX)}$ as follows: for any $\pi \in \Pi$, the state-measure $\mu_{\pi} \in \Lambda(\pi)$ is an invariant distribution of the transition probability $p(\,\cdot\,|x,\pi,\mu_{\pi})$, that is, 
\begin{align}
\mu_{\pi}(\,\cdot\,) = \sum_{x \in \sX} p(\,\cdot\,|x,a,\mu_{\pi}) \, \pi(a|x) \, \mu_{\pi}(x). \nonumber
\end{align}
Then, the notion of equilibrium for the MFG is defined as follows.

\tk{\begin{definition}
	A pair $(\pi_*,\mu_*) \in \Pi \times \Pnew(\sX)$ is called a mean-field equilibrium if $\pi_* \in \Psi(\mu_*)$ and $\mu_* \in \Lambda(\pi_*)$.
\end{definition}
}

Since IRL seeks to infer the reward function from expert demonstrations, it is important to impose some structure on the set of possible rewards to keep the problem manageable. As with MDPs, a common assumption in IRL for MFGs is that the reward can be represented as a linear combination of feature vectors, which depend on the state, action, and mean-field term:
$$
\cR \coloneqq \left\{r(x,a,\mu) = \langle \theta, f(x,a,\mu) \rangle: \theta \in \R^k, \,\, f:\sX\times\sA\times\Pnew(\sX) \rightarrow \R^k \right\}.
$$
In this context, $f(x,a,\mu) \in \R^k$ is the feature vector that captures information about the corresponding state $x$, action $a$, and the mean-field term $\mu$.

In the IRL setting, we suppose that an expert generates trajectories 
$$
\D=\left\{\left(x_i(t),a_i(t)\right)_{t\geq0}\right\}_{i=1}^d
$$
under some mean-field equilibrium $(\pi_E,\mu_E)$. Since $\mu_E$ is the stationary distribution of the transition probability under policy $\pi_E$ when the mean-field term in state dynamics is $\mu_E$, the ergodic theorem implies (as $\mu_E$ is the invariant distribution of the transition probability $p(\cdot|x,\pi_E,\mu_E)$) that
$$
\lim_{T\rightarrow \infty} \frac{1}{T} \sum_{t=0}^T \left( \frac{1}{d} \sum_{i=1}^d 1_{\{x_i(t)=x\}} \right) = \mu_E(x)
$$
for all $x \in \sX$. In the above limit, it is sufficient to consider only one sample path in $\D$. To obtain a more robust estimate of $\mu_E$, one can use all sample paths in $\D$. 
\tk{Furthermore, when $d$ is sufficiently large, we can approximate the feature expectation vector as
\begin{equation}\label{feat_exp_vec}
\frac{1}{d} \sum_{i=1}^d \left( \sum_{t=0}^{\infty} \beta^t \, f(x_i(t),a_i(t),\mu_E) \right) \simeq \langle f \rangle_{\pi_E,\mu_E},
\end{equation}
where $\langle f \rangle_{\pi_E,\mu_E} \coloneqq E^{\pi_E,\mu_E}\left[\sum_{t=0}^{\infty} \beta^t \,f(x(t),a(t),\mu_E)\right]$ represents the expected cumulative discounted feature value under the MFE $(\pi_E,\mu_E)$, with the initial distribution also given by $\mu_E$.
}

Therefore, in the remainder of this paper, we suppose that \tk{the} discounted feature expectation vector $\langle f \rangle_{\pi_E,\mu_E} $ under $(\pi_E,\mu_E)$ and the mean-field term $\mu_E$ are given.

\ns{
\begin{remark}
\tk{In the literature on IRL for MDPs, the feature expectation vector is commonly assumed to be known. Extending this to MFGs, recent studies \citep{YaLiLiHu22,ChZhLiWi23} demonstrate that an unbiased estimate of the mean-field term can be obtained from agent trajectories. By the ergodic theorem, this estimate converges to the true mean-field term with probability 1. Therefore, the assumption that the true mean-field term \(\mu_E\) is known does not impose a significant practical restriction.\footnote{\tk{Using estimates instead of true quantities produces an approximate MFE without fundamentally changing the problem's structure, although estimation errors may occur. Since the error analysis is straightforward and adds little insight, we do not explore it further.}}
}


\end{remark}}
We define the maximum discounted causal entropy IRL problem as follows:
\renewcommand\arraystretch{1.5}
\[
\begin{array}{lll}
\mathbf{(OPT_1)} \,\, \text{maximize}_{\pi} \text{ } &H(\pi) 
\\
\phantom{\mathbf{(OPT_1)} \,\, } \text{subject to} & \pi(a|x) \geq 0 \,\, \forall (x,a) \in \sX \times \sA   \\
\phantom{x} & \sum_{a \in \sA} \pi(a|x) = 1 \,\, \forall x \in \sX \\
\phantom{x} & \mu_E(x) = \sum_{(a,y) \in \sA \times \sX} p(x|y,a,\mu_E) \, \pi(a|y) \, \mu_E(y) \,\, \forall x \in \sX \\
 \phantom{x} & \sum_{t=0}^{\infty} \beta^t \, E^{\pi,\mu_E}[f(x(t),a(t),\mu_E)] = \langle f \rangle_{\pi_E,\mu_E},
\end{array}
\]
\tk{
where $$
H(\pi) = \sum_{t=0}^{\infty} \beta^t E^{\pi,\mu_E} \left[-\log \, \pi(a(t)|x(t)) \right]
$$
is the discounted causal entropy of the policy $\pi$.}
In this problem, the expert behaves according to some mean-field equilibrium $(\pi_E,\mu_E)$ under some unknown reward function $r_E(x,a,\mu) = \langle \theta_E,f(x,a,\mu) \rangle$. Therefore, $\pi_E$ is the optimal policy for $\mu_E$ under $r_E$. On the other hand, $\mu_E$ is the stationary distribution of the state under policy $\pi_E$ and the initial distribution $\mu_E$, when the mean-field term in state dynamics is $\mu_E$. Typically, there can be many  $\theta$ values that can explain this behavior, much like in the setting of MDPs. To address this inherent ambiguity, we employ the maximum causal entropy principle, which dictates that when confronted with multiple candidates explaining the behavior, one should select the one with the highest causal entropy. This allows us to avoid any bias except for the bias introduced by the feature expectation constraint.

\tk{Let $\pi^*$ be the solution of $\mathbf{(OPT_1)}$.  Then, we have
$$
E^{\pi^*,\mu_E}[f(x(t),a(t),\mu)] = \langle f \rangle_{\pi_E,\mu_E},
$$
and from \eqref{feat_exp_vec}, it follows that
$$
E^{\pi^*,\mu_E}[r_E(x(t),a(t),\mu)] = E^{\pi_E,\mu_E}[r_E(x(t),a(t),\mu_E)]. 
$$}
Therefore, $\pi^*$ is also an optimal policy for $\mu_E$ likewise $\pi_E$ (i.e., $\pi^* \in \Psi(\mu_E)$). From
$$
\mu_E(x) = \sum_{(a,y) \in \sA \times \sX} p(x|y,a,\mu_E) \, \pi^*(a|y) \, \mu_E(y) \,\, \forall x \in \sX,
$$
it follows that $\mu_E \in \Lambda(\pi^*)$; and therefore, $(\pi^*,\mu_E)$ is a mean-field equilibrium as well. Hence, solving $\mathbf{(OPT_1)}$ also leads to a MFE, similar to the MDP setting.

\begin{remark}
In $\mathbf{(OPT_1)}$, we can include the mean-field term $\mu$ as a variable in addition to $\pi$ if we suppose that $\mu_E$ is not available to us and obtain
\renewcommand\arraystretch{1.5}
\[
\begin{array}{lll}
\mathbf{(OPT_o)} \,\, \text{maximize}_{\pi,\mu} \text{ } &H(\pi,\mu) 
\\
\phantom{\mathbf{(OPT_1)} \,\, } \text{subject to} & \pi(a|x) \geq 0 \,\, \forall (x,a) \in \sX \times \sA   \\
\phantom{x} & \sum_{a \in \sA} \pi(a|x) = 1 \,\, \forall x \in \sX \\
\phantom{x} & \sum_{x \in \sX} \mu(x) = 1\\
\phantom{x} & \mu(x) \geq 0 \,\, \forall x \in \sX \\
\phantom{x} & \mu(x) = \sum_{(a,y) \in \sA \times \sX} p(x|y,a,\mu) \, \pi(a|y) \, \mu(y) \,\, \forall x \in \sX \\
 \phantom{x} & \sum_{t=0}^{\infty} \beta^t \, E^{\pi,\mu}[f(x(t),a(t),\mu)] = \langle f \rangle_{\pi_E,\mu_E},
\end{array}
\]
where $H(\pi,\mu)$ is defined as 
$$
H(\pi,\mu) \coloneqq \sum_{t=0}^{\infty} \beta^t E^{\pi,\mu} \left[-\log \, \pi(a(t)|x(t)) \right]. 
$$
If $(\pi^*,\mu^*)$ is the optimal solution to $\mathbf{(OPT_o)}$, a potential issue arises when $\mu^* \neq \mu_E$. In this scenario, the policy $\pi^*$ no longer qualifies as an optimal policy for either $\mu^*$ or $\mu_E$ because the condition $\sum_{t=0}^{\infty} \beta^t \, E^{\pi^*,\mu^*}[f(x(t),a(t),\mu^*)] = \langle f \rangle_{\pi_E,\mu_E}$ does not imply either of these optimality results. Consequently, $(\pi^*,\mu^*)$ and $(\pi^*,\mu_E)$ cannot be considered a mean-field equilibrium, which is an undesirable outcome.
\end{remark}

\tk{
\begin{remark}
Entropy-regularized MFGs, as developed in \citet{CuKo21,AnKaSa23r}, ensure the uniqueness of the MFE by incorporating an entropy term into the reward structure. This regularization smooths the optimization landscape, improves convergence, and ensures the uniqueness of equilibrium solutions by favoring balanced policies. However, it leads to an approximate MFE that may deviate from the true equilibrium of the original game. An alternative approach to achieving a unique MFE is the use of $\mathbf{(OPT_1)}$ which selects the MFE that maximizes entropy from among all possible equilibria corresponding to $\mu_E$.
In contrast to entropy-regularized MFGs, which yield an approximate MFE, the proposed method delivers an exact MFE that is characterized by the highest entropy. This results in a more stable and balanced solution akin to that provided by entropy regularization. Nevertheless, effective implementation of this approach necessitates the computation of the MFE and the generation of the feature expectation vector, or alternatively, the direct construction of this vector, which can be computationally demanding.
 \end{remark}
}

\subsection{Convexification of $\mathbf{(OPT_1)}$} 
 
Note that $\mathbf{(OPT_1)}$ is not a convex optimization problem as its last constraint, required for discounted feature expectation matching, is not convex in $\pi$. To convexify the problem, we employ normalized occupation measures induced by policies, a technique similarly used in \citet{ZhBlBa18} to address the IRL problem in the infinite-horizon MDPs. For any policy $\pi$, we define the state-action normalized occupation measure as 
$$
\nu_{\pi}(x,a) \coloneqq (1-\beta) \, \sum_{t=0}^{\infty} \beta^t \, E^{\pi,\mu_E} \left[ 1_{\{(x(t),a(t)) = (x,a)\}} \right]. 
$$
The constant factor $(1-\beta)$ in the definition makes $\nu_{\pi}$ a probability measure. \tk{Let us also define the corresponding marginal occupation measure as $\nu_{\pi}^{\sX}(x) \coloneqq \sum_{a \in \sA} \nu_{\pi}(x,a)$.} The following result establishes a flow condition satisfied by state-action normalized occupation measure under feasible policies of $\mathbf{(OPT_1)}$.

\begin{lemma}\label{a-lemma1}
\tk{If $\pi$ is a feasible point for $\mathbf{(OPT_1)}$, then 
	$\nu_{\pi}^{\sX}(x)= \mu_E(x)$ for every $x\in \sX$. Moreover, $\nu_{\pi}$ satisfies the following flow condition:
$$
\nu_{\pi}^{\sX}(x) = \sum_{(y,a) \in \sX \times \sA} p(x|y,a,\mu_E) \, \nu_{\pi}(y,a) \,\, \forall x \in \sX.
$$}
\end{lemma}
\begin{proof}
Without any constraint on $\pi$, this occupation measure satisfies the  Bellman flow condition
$$
\nu_{\pi}^{\sX}(x) = (1-\beta) \, \mu_E(x) + \beta \, \sum_{(y,a) \in \sX \times \sA} p(x|y,a,\mu_E) \, \nu_{\pi}(y,a) 
$$
for all $x \in \sX$. Note that $\nu_{\pi}$ can be disintegrated as 
$$
\nu_{\pi}(x,a) = \pi(a|x) \, \nu_{\pi}^{\sX}(x),
$$
where
$$
\nu_{\pi}^{\sX}(x) = (1-\beta) \, \sum_{t=0}^{\infty} \beta^t \, E^{\pi,\mu_E} \left[ 1_{\{x(t)= x\}} \right] = (1-\beta) \, \sum_{t=0}^{\infty} \beta^t \, {\Law}^{\pi,\mu_E}\{x(t)\}(x).
$$
If $\pi$ is a feasible point for $\mathbf{(OPT_1)}$, then it satisfies the following additional constraint
$$
\mu_E(x) = \sum_{(a,y) \in \sA \times \sX} p(x|y,a,\mu_E) \, \pi(a|y) \, \mu_E(y) \,\, \forall x \in \sX.
$$
Hence, we have ${\Law}^{\pi,\mu_E}\{x(t)\} = \mu_E$ for all $t\geq0$, as $x(0) \sim \mu_E$. This implies $\nu_{\pi}^{\sX}= \mu_E$, and the Bellman flow condition in this case can be written as
$$
\nu_{\pi}^{\sX}(x) = \sum_{(y,a) \in \sX \times \sA} p(x|y,a,\mu_E) \, \nu_{\pi}(y,a) \,\, \forall x \in \sX. 
$$
\end{proof}

Using Lemma~\ref{a-lemma1}, we can also prove the following.
\begin{lemma}\label{a-lemma2}
If $\pi$ is a feasible point for $\mathbf{(OPT_1)}$, then $H(\pi)$ can be written as
\begin{align*}
H(\pi) &= \frac{1}{1-\beta} \,  \sum_{(x,a) \in \sX \times \sA} -\log\left(\frac{\nu_{\pi}(x,a)}{\mu_E(x)}\right) \, \nu_{\pi}(x,a).
\end{align*}
\end{lemma}
\begin{proof}
\ns{
Recall that, for any policy $\pi$, we have
\begin{align*}
H(\pi) &= \sum_{t=0}^{\infty} \beta^t E^{\pi,\mu_E} \left[-\log \, \pi(a(t)|x(t)) \right] \\
&= \sum_{t=0}^{\infty} \beta^t \sum_{(x,a) \in \sX \times \sA} E^{\pi,\mu_E} \left[-\log \, \pi(a|x) \, 1_{\{(x(t),a(t)) = (x,a)\}} \right] \\
&= \sum_{(x,a) \in \sX \times \sA} \sum_{t=0}^{\infty} \beta^t \, E^{\pi,\mu_E} \left[-\log \, \pi(a|x) \, 1_{\{(x(t),a(t)) = (x,a)\}} \right] \\
&= \frac{1}{1-\beta} \,  \sum_{(x,a) \in \sX \times \sA} -\log\left(\pi(a|x)\right) \, \nu_{\pi}(x,a).
\end{align*}
}
Namely, one can write the discounted causal entropy of any policy $\pi$ as
$$
H(\pi) = \frac{1}{1-\beta} \,  \sum_{(x,a) \in \sX \times \sA} -\log\left(\pi(a|x)\right) \, \nu_{\pi}(x,a)
$$
using the occupation measure $\nu_{\pi}$. Note that we can disintegrate $\nu_{\pi}$ as $$\nu_{\pi}(x,a) = \pi(a|x) \, \nu_{\pi}^{\sX}(x).$$ Hence, we can alternatively write the discounted causal entropy of the policy $\pi$ as
$$
H(\pi) = \frac{1}{1-\beta} \,  \sum_{(x,a) \in \sX \times \sA} -\log\left(\frac{\nu_{\pi}(x,a)}{\nu_{\pi}^{\sX}(x)}\right) \, \nu_{\pi}(x,a). 
$$
This is true for any $\pi \in \Pi$. If $\pi$ is feasible for $\mathbf{(OPT_1)}$, then by Lemma~\ref{a-lemma1}, $\nu_{\pi}^{\sX}= \mu_E$. Hence, the discounted causal entropy of $\pi$ can also be expressed in the form of
\begin{align*}
H(\pi) &= \frac{1}{1-\beta} \,  \sum_{(x,a) \in \sX \times \sA} -\log\left(\frac{\nu_{\pi}(x,a)}{\mu_E(x)}\right) \, \nu_{\pi}(x,a).
\end{align*}
\end{proof}
In addition to Lemma~\ref{a-lemma2}, the discounted feature expectation vector can also be written as
\begin{align}\label{a-eq1}
\sum_{t=0}^{\infty} \beta^t \, E^{\pi,\mu_E}[f(x(t),a(t),\mu_E)] \coloneqq \langle f \rangle_{\pi,\mu_E} &= \frac{1}{1-\beta} \,  \sum_{(x,a) \in \sX \times \sA} f(x,a,\mu_E) \, \nu_{\pi}(x,a).
\end{align}
Consequently, we can define the following (convex) optimization problem, which will be proven to be equivalent to $\mathbf{(OPT_1)}$: 
\renewcommand\arraystretch{1.5}
\[\displaystyle
\begin{array}{lll}
\mathbf{(OPT_2)} \,\, \text{maximize}_{\nu} \text{ } &\frac{1}{1-\beta} \,  \sum_{(x,a) \in \sX \times \sA} -\log\left(\frac{\nu(x,a)}{\mu_E(x)}\right) \, \nu(x,a) 
\\
\phantom{\mathbf{(OPT_2)} \,\,} \text{subject to} & \frac{1}{1-\beta} \,  \sum_{(x,a) \in \sX \times \sA} f(x,a,\mu_E) \, \nu(x,a) = \langle f \rangle_{\pi_E,\mu_E}   \\
\phantom{z} & \mu_E(x) = \sum_{(y,a) \in \sX \times \sA} p(x|y,a,\mu_E) \, \nu(y,a)  \,\, \forall x \in \sX \\
\phantom{x} & \nu^{\sX}(x) = \mu_E(x) \,\, \forall x \in \sX \\
 \phantom{x} & \nu(x,a) \geq 0 \,\, \forall (x,a) \in \sX \times \sA.
\end{array}
\]
\tk{Here, $\nu^{\sX}(x) \coloneqq \sum_{a \in \sA} \nu(x,a)$. Instead of formulating the problem over the set of policies, we express it with respect to the set of occupation measures induced by policies. This reformulation convexifies the problem, as demonstrated below.}
\begin{lemma}\label{a-lemma3}
The optimization problem $\mathbf{(OPT_2)}$ is convex.
\end{lemma}
\begin{proof}
Firstly, the equality and inequality constraints in $\mathbf{(OPT_2)}$ are all linear in $\nu$. Moreover, the objective function is strongly concave in $\nu$. Indeed, rewriting the objective function as
$$
\frac{1}{1-\beta} \left(  \sum_{(x,a) \in \sX \times \sA} -\log\left(\nu(x,a)\right) \, \nu(x,a) + \sum_{(x,a) \in \sX \times \sA} \log\left(\mu_E(x)\right) \, \nu(x,a) \right),
$$
the first term in above sum is the entropy of the distribution $\nu$, which is known to be strongly concave in $\nu$, and the second term is linear in $\nu$. Hence, the objective function is strongly concave overall.
\end{proof} 
Now it is time to prove the equivalence of $\mathbf{(OPT_1)}$ and $\mathbf{(OPT_2)}$.

\begin{theorem}\label{equiv_thm}
The optimization problems $\mathbf{(OPT_1)}$ and $\mathbf{(OPT_2)}$ are equivalent; that is, there are mappings between feasible points and two equivalent feasible points under these mappings lead to the same objective value. 
\end{theorem}

\begin{proof}
Let $\pi$ be a feasible point for $\mathbf{(OPT_1)}$. Then, consider the corresponding occupation measure $\nu_{\pi}$. By Lemma~\ref{a-lemma1}, 
\tk{it follows that
$$
\mu_E(x) = \sum_{(y,a) \in \sX \times \sA} p(x|y,a,\mu_E) \, \nu_{\pi}(y,a),  \,\,  \forall x \in \sX.
$$
}
Moreover, from (\ref{a-eq1}), we also have 
$$
\frac{1}{1-\beta} \sum_{(x,a) \in \sX \times \sA} f(x,a,\mu_E) \, \nu_{\pi}(x,a) = \langle f \rangle_{\pi_E,\mu_E}.
$$
Therefore, $\nu_{\pi}$ is feasible for $\mathbf{(OPT_2)}$. On the other hand, from Lemma~\ref{a-lemma2} we obtain
\begin{align*}
H(\pi) &= \frac{1}{1-\beta} \,  \sum_{(x,a) \in \sX \times \sA} -\log\left(\frac{\nu_{\pi}(x,a)}{\mu_E(x)}\right) \, \nu_{\pi}(x,a),
\end{align*}
indicating that the objectives of $\pi$ and $\nu_{\pi}$ are equivalent. This completes the proof of the ``only if" part.

Conversely, let $\nu$ be a feasible point for $\mathbf{(OPT_2)}$. Then, define 
\begin{align*}
\pi_{\nu}(a|x) = \begin{cases}
\frac{\nu(x,a)}{\nu^{\sX}(x)} & \text{if} \,\, \nu^{\sX}(x) > 0, \\
\pi(a) & \text{if} \,\, \nu^{\sX}(x) = 0,
\end{cases}
\end{align*}
where $\pi \in \Pnew(\sX)$ is arbitrary.
Consider the occupation measure $\nu_{\pi_{\nu}}$ induced by $\pi_{\nu}$. Since  
\begin{align*}
\nu^{\sX}(x) &= \sum_{(y,a) \in \sX \times \sA} p(x|y,a,\mu_E) \, \pi_{\nu}(a|y) \, \nu^{\sX}(y)  \,\, \forall x \in \sX \\
\nu^{\sX}(x) &= \mu_E(x) \,\, \forall x \in \sX, \,\,\,\,
x(0)  \sim \mu_E,
\end{align*}
we have $\Law^{\pi_{\nu},\mu_E}\{x(t)\} = \nu^{\sX}$ for all $t \geq 0$. Also, since 
$$
\nu_{\pi_{\nu}}^{\sX} = (1-\beta) \, \sum_{t=0}^{\infty} \beta^t \, {\Law}^{\pi_{\nu},\mu_E}\{x(t)\},
$$ 
we have $\nu_{\pi_{\nu}}^{\sX} = \nu^{\sX} = \mu_E$. This implies that 
$$
\nu_{\pi_{\nu}}(x,a) = \pi_{\nu}(a|x) \, \nu_{\pi_{\nu}}^{\sX}(x) = \pi_{\nu}(a|x) \,\nu^{\sX}(x) = \nu(x,a) \,\, \forall (x,a) \in \sX\times\sA.
$$
This completes the proof of the ``if" part in view of the following:
\begin{align*}
H(\pi_{\nu}) &= \frac{1}{1-\beta} \,  \sum_{(x,a) \in \sX \times \sA} -\log\left(\frac{\nu_{\pi_{\nu}}(x,a)}{\mu_E(x)}\right) \, \nu_{\pi_{\nu}}(x,a)  & \text{(Lemma~\ref{a-lemma2})} \\
\sum_{t=0}^{\infty} \beta^t \, E^{\pi_{\nu},\mu_E}[f(x(t),a(t),\mu_E)] &= \frac{1}{1-\beta} \,  \sum_{(x,a) \in \sX \times \sA} f(x,a,\mu_E) \, \nu_{\pi_{\nu}}(x,a)  & \text{(Eq. (\ref{a-eq1}))}.
\end{align*}
\end{proof}

\begin{remark}
The relation between feasible points of $\mathbf{(OPT_1)}$ and $\mathbf{(OPT_2)}$ is bijective if $\mu_E(x) > 0$ for all $x \in \sX$. Indeed, we have 
\begin{align*}
\Pi \ni \pi &\mapsto \nu_{\pi} \eqqcolon (1-\beta) \, \sum_{t=0}^{\infty} \beta^t \, E^{\pi,\mu_E} \left[ 1_{\{(x(t),a(t)) \eqqcolon (x,a)\}} \right] \in \Pnew(\sX\times\sA) \\
\Pnew(\sX\times\sA) \ni \nu &\mapsto \pi_{\nu} = \frac{\nu(x,a)}{\nu^{\sX}(x)} \in \Pi.
\end{align*}
In both cases, we have $\nu_{\pi}^{\sX}(x) = \nu^{\sX}(x) = \mu_E(x) > 0$ for all $x \in \sX$ by assumption. This implies that the second mapping above is injective; and so, above mappings define a bijective relation between feasible points of $\mathbf{(OPT_1)}$ and $\mathbf{(OPT_2)}$. 
\end{remark}

\subsection{Dual of $\mathbf{(OPT_2)}$ and Its Solution}

In the remainder of this section, our aim is to establish an algorithm to compute the optimal solution of $\mathbf{(OPT_2)}$, as it is equivalent to $\mathbf{(OPT_1)}$. To this end, we first prove the following result.
\begin{proposition}\label{a-prop1}
The dual of $\mathbf{(OPT_2)}$ is 
$$
\min_{\btheta \in \R^k, \,  \blambda, \bxi \in \R^{\sX}} \bigg\{ \frac{1}{1-\beta} \, \log \sum_{(x,a) \in \sX \times \sA} e^{k_{\btheta,\blambda,\bxi}(x,a)}- \langle \btheta,\langle f \rangle_{\pi_E,\mu_E} \rangle- \sum_{x \in \sX} \blambda_x \, \mu_E(x) \bigg\},
$$
where 
$$
k_{\btheta,\blambda,\bxi}(x,a) \coloneqq \log \mu_E(x) + \langle \btheta, f(x,a,\mu_E) \rangle + (1-\beta) \left[ \blambda_x + \sum_{z \in \sX} \bxi_z \,  \left( p(z|x,a,\mu_E) - \mu_E(z) \right) \right].
$$
Moreover, the optimal value of the dual problem is equal to the optimal value of $\mathbf{(OPT_2)}$; that is, there is no duality gap. Finally, if $(\btheta^*,\blambda^*,\bxi^*)$ is the optimal solution of the dual problem, then the optimal solution of $\mathbf{(OPT_2)}$ is the following Boltzman distribution
$$
\nu^*_{\btheta^*,\blambda^*,\bxi^*}(x,a) \coloneqq \frac{e^{k_{\btheta^*,\blambda^*,\bxi^*}(x,a) }}{\sum_{(x,a) \in \sX \times \sA} e^{k_{\btheta^*,\blambda^*,\bxi^*}(x,a) }}.
$$
\end{proposition}
\begin{proof}
To obtain the dual form of $\mathbf{(OPT_2)}$, let us introduce the problem $\mathbf{(OPT_2)}$ in a $\min$-$\max$ formulation:
\begin{multline*}
\mathbf{(OPT_2)} 
= \max_{\nu \in \Pnew(\sX \times \sA)} \min_{\btheta \in \R^k, \,  \blambda, \bxi \in \R^{\sX}} \frac{1}{1-\beta} \left[ H(\nu) +  \sum_{(x,a) \in \sX \times \sA} k_{\btheta,\blambda,\bxi}(x,a) \, \nu(x,a) \right]  \\- \langle \btheta,\langle f \rangle_{\pi_E,\mu_E} \rangle - \sum_{x \in \sX} \blambda_x \, \mu_E(x).
\end{multline*}
Then, according to Sion's minimax theorem \citep[see][]{Sio58}, we can interchange the minimum and maximum in the expression above, leading to the dual form:
\begin{multline*}
\mathbf{(OPT_2)} 
= \min_{\btheta \in \R^k, \,  \blambda, \bxi \in \R^{\sX}} \max_{\nu \in \Pnew(\sX \times \sA)}  \frac{1}{1-\beta} \left[ H(\nu) +  \sum_{(x,a) \in \sX \times \sA} k_{\btheta,\blambda,\bxi}(x,a) \, \nu(x,a) \right] \\ - \langle \btheta,\langle f \rangle_{\pi_E,\mu_E} \rangle - \sum_{x \in \sX} \blambda_x \, \mu_E(x) \\
= \min_{\btheta \in \R^k, \,  \blambda, \bxi \in \R^{\sX}} \bigg\{ \frac{1}{1-\beta} \,  \max_{\nu \in \Pnew(\sX \times \sA)} \left[ H(\nu) +  \sum_{(x,a) \in \sX \times \sA} k_{\btheta,\blambda,\bxi}(x,a) \, \nu(x,a) \right] \\ - \langle \btheta,\langle f \rangle_{\pi_E,\mu_E} \rangle - \sum_{x \in \sX} \blambda_x \, \mu_E(x) \bigg\} \\
= \min_{\btheta \in \R^k, \,  \blambda, \bxi \in \R^{\sX}} \bigg\{ \frac{1}{1-\beta} \, \log \sum_{(x,a) \in \sX \times \sA} e^{k_{\btheta,\blambda,\bxi}(x,a)}- \langle \btheta,\langle f \rangle_{\pi_E,\mu_E} \rangle- \sum_{x \in \sX} \blambda_x \, \mu_E(x) \bigg\}. 
\end{multline*}
Here, the last equality follows from the variational formula
$$
\log \sum_{z \in \sZ}e^{k(z)} = \max_{\nu \in \Pnew(\sZ)} \left[H(\nu) + \sum_{z \in \sZ} k(z) \, \nu(z) \right]
$$
\citep[see][Proposition 1.4.2]{DuEl97}.\footnote{Normally, in large deviation theory, the variational formula is formulated using relative entropy. However, for finite spaces, the above result can be obtained by considering the relationship between the entropy of a distribution and the relative entropy of that distribution with respect to the uniform distribution.}  This implies that above minimization problem is the dual of $\mathbf{(OPT_2)}$ and by Sion's minimax theorem, there is no duality gap. Moreover, for a given $\btheta \in \R^k, \,  \blambda, \bxi \in \R^{\sX}$,  the probability measure $\nu^*_{\btheta,\blambda,\bxi}$ that maximizes 
$$H(\nu) +  \sum_{(x,a) \in \sX \times \sA} k_{\btheta,\blambda,\bxi}(x,a) \, \nu(x,a)$$ 
is the Boltzman distribution defined by \tk{
$$
\nu^*_{\btheta,\blambda,\bxi}(x,a) \coloneqq \frac{e^{k_{\btheta,\blambda,\bxi}(x,a) }}{Z_{\btheta,\blambda,\bxi}},
$$
where
$
Z_{\btheta,\blambda,\bxi} = \sum_{(x,a) \in \sX \times \sA} e^{k_{\btheta,\blambda,\bxi}(x,a)}
$} \citep[Proposition 1.4.2]{DuEl97}. This completes the proof of the last assertion.
\end{proof}

We denote the objective function in dual formulation of $\mathbf{(OPT_2)}$ as
\begin{align*}
&g(\btheta,\blambda, \bxi) \coloneqq \frac{1}{1-\beta} \, \log \sum_{(x,a) \in \sX \times \sA} e^{k_{\btheta,\blambda,\bxi}(x,a)} - \langle \btheta,\langle f \rangle_{\pi_E,\mu_E} \rangle - \sum_{x \in \sX} \blambda_x \, \mu_E(x).
\end{align*}
Note that since $H(\nu) +  \sum_{(x,a) \in \sX \times \sA} k_{\btheta,\blambda,\bxi}(x,a) \, \nu(x,a)$ is linear in $(\btheta,\blambda,\bxi)$, its maximum over $\nu$, which is given by  
$$ \log \sum_{(x,a) \in \sX \times \sA} e^{k_{\btheta,\blambda,\bxi}(x,a)},$$ 
is a convex function of $(\btheta,\blambda,\bxi)$.  Consequently, $g$ is a convex function as the remaining terms in $g$ are linear in $(\btheta,\blambda,\bxi)$. Hence,
the dual formulation of $\mathbf{(OPT_2)}$ is an unconstrained convex optimization problem in the variables $(\btheta,\blambda,\bxi)$. Therefore, it can be solved using a gradient descent algorithm, which we will proceed to describe. Before introducing the algorithm, we first establish the smoothness and strong convexity of $g$. This guarantees convergence of the gradient descent algorithm with an explicit rate under a constant step-size. To obtain the strong convexity of $g$, we need to make the following assumption. 

\begin{assumption}\label{strong-convexity}
$$
\cspan \bigg\{ \bigg(f(x,a,\mu_E),p(\cdot|x,a,\mu_E),e(\cdot|x,a)\bigg) : (x,a) \in \sX \times \sA \bigg\} = \R^k \times \R^{\sX} \times \R^{\sX},
$$
where 
\begin{align*}
e(y|x,a) =
\begin{cases}
1, &  \text{if} \,\, x = y \\
0, &\text{otherwise.}
\end{cases}
\end{align*}
\end{assumption} 

\ns{
\begin{remark}
Assumption~\ref{strong-convexity} is not particularly restrictive, and it is easy to construct an example that satisfies this condition. To keep the discussion simple, let us consider the following scenario. Define the transition probability function as
\[
p(y|x,a,\mu_E) =
\begin{cases}
1, & \text{if} \,\, l(x) = y, \\
0, & \text{otherwise},
\end{cases}
\]
where \( l: \sX \rightarrow \sX \) is a bijective mapping. Additionally, define the function
\[
f_{y,b}(x,a,\mu_E) =
\begin{cases}
1, & \text{if} \,\, h(x,a) = (y,b), \\
0, & \text{otherwise},
\end{cases}
\]
where \( f(x,a,\mu_E) = (f_{y,b}(x,a,\mu_E))_{(y,b) \in \sX \times \sA} \) and \( h: \sX\times\sA \rightarrow \sX\times\sA \) is also bijective. Thus, in this particular example, \( k = |\sX| + |\sA| \). Now, let \( u = (u_1,u_2,u_3) \in \R^k \times \R^{\sX} \times \R^{\sX} \) and consider the inner product of \( (u_1,u_2,u_3) \) with $\big(f(x,a,\mu_E),p(\cdot|x,a,\mu_E),e(\cdot|x,a)\big)$:
\begin{align*}
&\sum_{(y,b) \in \sX \times \sA} f_{y,b}(x,a,\mu_E) \, u_1(y,b) + \sum_{y \in \sX} p(y|x,a,\mu_E) \, u_2(y) + \sum_{y \in \sX} e(y|x,a) \, u_3(y) \\
&\phantom{xxxxxxxxxxxxxxxxxxx}= u_1(h(x,a)) + u_2(l(x)) + u_3(x).
\end{align*}
If this inner product equals zero for all \( (x,a) \in \sX \times \sA \), then, since \( l \) and \( h \) are bijections, the vector \( u \) must be zero. Therefore, Assumption~\ref{strong-convexity} holds for this example.
\end{remark}
}

Now, we can state the following theorem about smoothness and strong convexity of the objective function $g$.

\begin{theorem}
The objective function $g$ is $L$-smooth where 
$$
L \coloneqq 2M \, \left(\frac{M_1}{1-\beta} + 2 \, \sqrt{|\sX| \, |\sA|} \right)
$$ and the constants $M_1$ and $M$ can be determined explicitly. Moreover, under Assumption~\ref{strong-convexity}, for any convex compact subset $D$ in $\R^m$, $g$ is also $\rho(D)$-strongly convex on $D$ for some $\rho(D)$ that depends on $D$. 
\end{theorem}

\begin{proof}
Note that the partial gradients of $g$ with respect to the vectors $\btheta,\blambda, \bxi$ are given as follows:
\begin{align*}
\nabla_{\btheta} g(\btheta,\blambda, \bxi) &= \frac{1}{1-\beta} \,  \sum_{(x,a) \in \sX \times \sA} f(x,a,\mu_E) \, \nu^*_{\btheta,\blambda,\bxi}(x,a) - \langle f \rangle_{\pi_E,\mu_E} ,  \\
\nabla_{\bxi} g(\btheta,\blambda, \bxi) &= \sum_{(x,a) \in \sX \times \sA} p(\cdot|x,a,\mu_E) \, \nu^*_{\btheta,\blambda,\bxi}(x,a) -\mu_E(\cdot),  \\
\nabla_{\blambda} g(\btheta,\blambda, \bxi) &= \sum_{(x,a) \in \sX \times \sA} e(\cdot|x,a) \, \nu^*_{\btheta,\blambda,\bxi}(x,a) -\mu_E(\cdot) 
=\nu_{\btheta,\blambda,\bxi}^{*,\sX}(\cdot) - \mu_E(\cdot). 
\end{align*}
Therefore, to establish the Lipschitz continuity of $\nabla g$ (or, equivalently  smoothness of $g$), we need to first prove Lipschitz continuity of $\nu^*_{\btheta,\blambda,\bxi}$ with respect to $(\btheta,\blambda,\bxi)$. The partial gradients of $\nu^*_{\btheta,\blambda,\bxi}(x,a)$ with respect to the vectors $\btheta,\blambda, \bxi$ are given by
\begin{align*}
&\nabla_{\bepsilon} \nu^*_{\btheta,\blambda,\bxi}(x,a) \\
&= \frac{e^{k_{\btheta,\blambda,\bxi}}(x,a) \, Z_{\btheta,\blambda,\bxi} \, \nabla_{\bepsilon} k_{\btheta,\blambda,\bxi}(x,a) - e^{k_{\btheta,\blambda,\bxi}}(x,a) \, \sum_{(y,b) \in \sX \times \sA} e^{k_{\btheta,\blambda,\bxi}}(y,b) \, \nabla_{\bepsilon} k_{\btheta,\blambda,\bxi}(y,b)}{(Z_{\btheta,\blambda,\bxi})^2} \\
&= \nu^*_{\btheta,\blambda,\bxi}(x,a) \, \nabla_{\bepsilon} k_{\btheta,\blambda,\bxi}(x,a) - \nu^*_{\btheta,\blambda,\bxi}(x,a) \, \langle \nabla_{\bepsilon} k_{\btheta,\blambda,\bxi} \rangle_{\nu^*_{\btheta,\blambda,\bxi}},
\end{align*}
where $\bepsilon \in \{\btheta,\blambda,\bxi\}$. Note that we have 
\begin{align*}
\sup_{\substack{(x,a) \in \sX \times \sA \\ \btheta,\blambda,\bxi}} \|\nabla_{\btheta} k_{\btheta,\blambda,\bxi}(x,a)\| &= \sup_{\substack{(x,a) \in \sX \times \sA \\ \btheta,\blambda,\bxi}} \|f(x,a,\mu_E)\| \eqqcolon M_1 < \infty, \\
\sup_{\substack{(x,a) \in \sX \times \sA \\ \btheta,\blambda,\bxi}} \|\nabla_{\blambda} k_{\btheta,\blambda,\bxi}(x,a)\| &= \sup_{\substack{(x,a) \in \sX \times \sA \\ \btheta,\blambda,\bxi}} (1-\beta) \|{\ee}_{x}\| \eqqcolon M_2 < \infty, \\
\sup_{\substack{(x,a) \in \sX \times \sA \\ \btheta,\blambda,\bxi}} \|\nabla_{\bxi} k_{\btheta,\blambda,\bxi}(x,a)\| &= \sup_{\substack{(x,a) \in \sX \times \sA \\ \btheta,\blambda,\bxi}} (1-\beta) \|p(\cdot|x,a,\mu_E)-\mu_E(\cdot)\| \eqqcolon M_3 < \infty ,
\end{align*}
where ${\ee}_{x} \in \R^{|\sX|}$ is the vector whose $x^{th}$ term is 1 and the rest are 0. This implies that for every $\bepsilon \in \{\btheta,\blambda,\bxi\}$ we have
\begin{align*}
\sup_{\substack{(x,a) \in \sX \times \sA \\ \btheta,\blambda,\bxi}}\|\nabla_{\bepsilon} \nu^*_{\btheta,\blambda,\bxi}(x,a)\| \leq 2 \, \max\{M_1,M_2,M_3\} \eqqcolon 2 \, M.
\end{align*}
Hence, by the mean-value theorem, $\nu^*_{\btheta,\blambda,\bxi}(x,a)$ is $2M$-Lipschitz continuous with respect to $(\btheta,\blambda,\bxi)$ for all $(x,a) \in \sX \times \sA$. This implies that for any $(\btheta,\blambda,\bxi)$ and $(\btheta',\blambda',\bxi')$,  we have
\begin{align*}
\|\nabla_{\btheta} g(\btheta,\blambda, \bxi) - \nabla_{\btheta} g(\btheta',\blambda', \bxi') \| &\leq \frac{M_1}{1-\beta} \, 2M \,  \|(\btheta,\blambda, \bxi) -(\btheta',\blambda', \bxi') \|,  \\
\|\nabla_{\bxi} g(\btheta,\blambda, \bxi)-\nabla_{\bxi} g(\btheta',\blambda', \bxi')\| &\leq 2M \, \sqrt{|\sX| \, |\sA|} \,  \|(\btheta,\blambda, \bxi) -(\btheta',\blambda', \bxi') \|,  \\
\|\nabla_{\blambda} g(\btheta,\blambda, \bxi)-\nabla_{\blambda} g(\btheta',\blambda', \bxi') \| &\leq  2M \, \sqrt{|\sX| \, |\sA|} \,  \|(\btheta,\blambda, \bxi) -(\btheta',\blambda', \bxi') \|. 
\end{align*}
Hence $g(\btheta,\blambda,\bxi)$ is $L$-smooth with respect to $(\btheta,\blambda,\bxi)$, where 
$$
L \coloneqq 2M \, \left(\frac{M_1}{1-\beta} + 2 \, \sqrt{|\sX| \, |\sA|} \right).
$$

Now it is time to prove the $\rho$-strong convexity of $g(\btheta,\blambda,\bxi)$ over compact subsets. To this end, let us compute the partial Hessian of this function with respect to variables $\btheta,\blambda,\bxi$, building upon the previously derived results. Denoting the outer (or tensor) product of two vectors by $\otimes$, we have
\begin{align*}
&\nabla_{\btheta,\btheta}^2 g(\btheta,\blambda, \bxi) = \frac{1}{1-\beta} \,  \sum_{(x,a) \in \sX \times \sA} f(x,a,\mu_E) \otimes \nabla_{\btheta}\nu^*_{\btheta,\blambda,\bxi}(x,a) \\
&=\frac{1}{1-\beta} \,  \sum_{(x,a) \in \sX \times \sA} f(x,a,\mu_E) \otimes \left[ 
 \nu^*_{\btheta,\blambda,\bxi}(x,a) \, \nabla_{\btheta} k_{\btheta,\blambda,\bxi}(x,a) - \nu^*_{\btheta,\blambda,\bxi}(x,a) \, \langle \nabla_{\btheta} k_{\btheta,\blambda,\bxi} \rangle_{\nu^*_{\btheta,\blambda,\bxi}} \right] \\
&= \frac{1}{1-\beta} \,  \sum_{(x,a) \in \sX \times \sA} f(x,a,\mu_E) \otimes \left[ 
 \nu^*_{\btheta,\blambda,\bxi}(x,a) \, f(x,a,\mu_E) - \nu^*_{\btheta,\blambda,\bxi}(x,a) \, \langle f(x,a,\mu_E) \rangle_{\nu^*_{\btheta,\blambda,\bxi}} \right] \\
&= \begin{aligned}[t] \frac{1}{1-\beta} \,  \sum_{(x,a) \in \sX \times \sA} f(x,a,\mu_E) \otimes  &f(x,a,\mu_E) \nu^*_{\btheta,\blambda,\bxi}(x,a)\\
& - \left(\sum_{(x,a) \in \sX \times \sA} f(x,a,\mu_E) \nu^*_{\btheta,\blambda,\bxi}(x,a) \right)  \otimes \langle f(x,a,\mu_E)  \rangle_{\nu^*_{\btheta,\blambda,\bxi}} \end{aligned} \\
&= \begin{aligned}[t] \frac{1}{1-\beta} \,  \sum_{(x,a) \in \sX \times \sA} f(x,a,\mu_E) \otimes  &f(x,a,\mu_E) \nu^*_{\btheta,\blambda,\bxi}(x,a) \\
&- \langle f(x,a,\mu_E) \rangle_{\nu^*_{\btheta,\blambda,\bxi}}  \otimes \langle f(x,a,\mu_E)  \rangle_{\nu^*_{\btheta,\blambda,\bxi}}.  \end{aligned}
\end{align*}
Define the random vector $\X_f$ on the discrete probability space $(\sX\times\sA,\nu^*_{\btheta,\blambda,\bxi})$ as 
$$
\X_f(x,a) \coloneqq \frac{1}{1-\beta} \, f(x,a,\mu_E) \in \R^k.
$$
Then, the computations above imply that 
\begin{align*}
\nabla_{\btheta,\btheta}^2 g(\btheta,\blambda, \bxi) = E[ \X_f \otimes \X_f ] - E [\X_f] \otimes E[\X_f] = \Cov(\X_f).
\end{align*}
Similarly, we have 
\begin{align*}
&\nabla_{\bxi,\bxi}^2 g(\btheta,\blambda, \bxi) =  \sum_{(x,a) \in \sX \times \sA} p(\cdot|x,a,\mu_E) \otimes \nabla_{\bxi}\nu^*_{\btheta,\blambda,\bxi}(x,a) \\
&= \sum_{(x,a) \in \sX \times \sA} p(\cdot|x,a,\mu_E) \otimes \left[ 
 \nu^*_{\btheta,\blambda,\bxi}(x,a) \, \nabla_{\bxi} k_{\btheta,\blambda,\bxi}(x,a) - \nu^*_{\btheta,\blambda,\bxi}(x,a) \, \langle \nabla_{\bxi} k_{\btheta,\blambda,\bxi} \rangle_{\nu^*_{\btheta,\blambda,\bxi}} \right] \\
&=  \sum_{(x,a) \in \sX \times \sA} p(\cdot|x,a,\mu_E)  \otimes \left[ 
 \nu^*_{\btheta,\blambda,\bxi}(x,a) \, p(\cdot|x,a,\mu_E)  - \nu^*_{\btheta,\blambda,\bxi}(x,a) \, \langle p(\cdot|x,a,\mu_E)  \rangle_{\nu^*_{\btheta,\blambda,\bxi}} \right] \\
&= \begin{aligned}[t] \sum_{(x,a) \in \sX \times \sA} p(\cdot|x,a,\mu_E)  \otimes  &p(\cdot|x,a,\mu_E) \nu^*_{\btheta,\blambda,\bxi}(x,a) \\
&- \left(\sum_{(x,a) \in \sX \times \sA} p(\cdot|x,a,\mu_E)  \nu^*_{\btheta,\blambda,\bxi}(x,a) \right)  \otimes \langle p(\cdot|x,a,\mu_E)  \rangle_{\nu^*_{\btheta,\blambda,\bxi}} \end{aligned}
\\
&= \begin{aligned}[t] \sum_{(x,a) \in \sX \times \sA} p(\cdot|x,a,\mu_E)  \otimes  &p(\cdot|x,a,\mu_E)  \nu^*_{\btheta,\blambda,\bxi}(x,a) \\
 &- \langle p(\cdot|x,a,\mu_E)  \rangle_{\nu^*_{\btheta,\blambda,\bxi}}  \otimes \langle p(\cdot|x,a,\mu_E)  \rangle_{\nu^*_{\btheta,\blambda,\bxi}}. \end{aligned}
\end{align*}
Defining the random vector $\X_p$ on the discrete probability space $(\sX\times\sA,\nu^*_{\btheta,\blambda,\bxi})$ as
$$
\X_p(x,a) \coloneqq p(\cdot|x,a,\mu_E) \in \R^{\sX},
$$
we obtain 
\begin{align*}
\nabla_{\bxi,\bxi}^2 g(\btheta,\blambda, \bxi) = E[ \X_p \otimes \X_p ] - E[\X_p] \otimes E[\X_p] = \Cov(\X_p).
\end{align*}
Finally, we have 
\begin{align*}
&\nabla_{\blambda,\blambda}^2 g(\btheta,\blambda, \bxi) =  \sum_{(x,a) \in \sX \times \sA} e(\cdot|x,a) \otimes \nabla_{\blambda}\nu^*_{\btheta,\blambda,\bxi}(x,a) \\
&= \sum_{(x,a) \in \sX \times \sA} e(\cdot|x,a) \otimes \left[ 
 \nu^*_{\btheta,\blambda,\bxi}(x,a) \, \nabla_{\blambda} k_{\btheta,\blambda,\bxi}(x,a) - \nu^*_{\btheta,\blambda,\bxi}(x,a) \, \langle \nabla_{\blambda} k_{\btheta,\blambda,\bxi} \rangle_{\nu^*_{\btheta,\blambda,\bxi}} \right] \\
&=  \sum_{(x,a) \in \sX \times \sA} e(\cdot|x,a) \otimes \left[ 
 \nu^*_{\btheta,\blambda,\bxi}(x,a) \, e(\cdot|x,a)  - \nu^*_{\btheta,\blambda,\bxi}(x,a) \, \langle e(\cdot|x,a)  \rangle_{\nu^*_{\btheta,\blambda,\bxi}} \right] \\
&= \begin{aligned}[t] \sum_{(x,a) \in \sX \times \sA} e(\cdot|x,a)  \otimes  &e(\cdot|x,a) \nu^*_{\btheta,\blambda,\bxi}(x,a) \\
&- \left(\sum_{(x,a) \in \sX \times \sA} e(\cdot|x,a) \nu^*_{\btheta,\blambda,\bxi}(x,a) \right)  \otimes \langle e(\cdot|x,a)  \rangle_{\nu^*_{\btheta,\blambda,\bxi}} \end{aligned} \\
&= \sum_{(x,a) \in \sX \times \sA} e(\cdot|x,a)  \otimes  e(\cdot|x,a) \nu^*_{\btheta,\blambda,\bxi}(x,a) - \langle e(\cdot|x,a) \rangle_{\nu^*_{\btheta,\blambda,\bxi}}  \otimes \langle e(\cdot|x,a) \rangle_{\nu^*_{\btheta,\blambda,\bxi}}.
\end{align*}
Defining another random vector $\X_e$ on the discrete probability space $(\sX\times\sA,\nu^*_{\btheta,\blambda,\bxi})$ as 
$$
\X_e(x,a) \coloneq e(\cdot|x,a) \in \R^{\sX},
$$
we similarly obtain
\begin{align*}
\nabla_{\blambda,\blambda}^2 g(\btheta,\blambda, \bxi) = E[ \X_e \otimes \X_e ] - E[\X_e] \otimes E[\X_e] = \Cov(\X_e).
\end{align*}
Note that one can also compute the cross terms similarly and obtain 
\begin{align*}
\nabla_{\bxi,\btheta}^2 g(\btheta,\blambda, \bxi) &= E[ \X_f \otimes \X_p ] - E[\X_f] \otimes E[\X_p], \\
\nabla_{\blambda,\btheta}^2 g(\btheta,\blambda, \bxi) &= E[ \X_f \otimes \X_e ] - E[\X_f] \otimes E[\X_e], \\
\nabla_{\bxi,\blambda}^2 g(\btheta,\blambda, \bxi) &= E[ \X_p \otimes \X_e ] - E[\X_p] \otimes E[\X_e].
\end{align*}
Hence, if we define the random vector ${\boldsymbol \X} \coloneqq (\X_f,\X_p,\X_e)$, then the Hessian of $g(\btheta,\blambda,\bxi)$ can be written as
$$
\Hes(g)(\btheta,\blambda,\bxi) = \Cov({\boldsymbol \X}).
$$
Clearly, $\Cov({\boldsymbol \X})$ is dependent on the parameters $(\btheta, \blambda, \bxi)$. Moreover, each element of $\Cov({\boldsymbol \X})$ represents an expectation of a random variable with respect to the Boltzmann distribution $\nu^*_{\btheta, \blambda, \bxi}$. Since $\nu^{*}_{\btheta, \blambda, \bxi}(x, a)$ has been demonstrated to be $2M$-Lipschitz continuous with respect to $(\btheta, \blambda, \bxi)$ for all $(x, a) \in \sX \times \sA$, it is evident that $\Cov({\boldsymbol \X})$ is continuous in terms of $(\btheta, \blambda, \bxi)$.

Furthermore, as covariance matrices are inherently symmetric and positive semi-definite, it follows that  $\Cov({\boldsymbol \X})$ is positive semi-definite for any given $(\btheta, \blambda, \bxi)$. However, to establish its positive definiteness, Assumption~\ref{strong-convexity} is required.

Under this assumption, suppose in contrary that, $\Cov({\boldsymbol \X})$ is not positive definite. Then, there exists a vector $\ba \in \R^k \times \R^{\sX} \times \R^{\sX} \eqqcolon \R^m$ such that $\langle \ba, \Cov({\boldsymbol \X}) \, \ba \rangle = 0$; that is, if ${\boldsymbol \X} = (\X_i)_{i=1}^m$, then
\begin{align*}
0 = \sum_{i,j=1}^m a_j \Cov(\X_j,\X_i) a_i = \Var\left(\sum_{i=1}^m a_i \, \X_i \right).
\end{align*}
This implies that the random variable $\sum_{i=1}^m a_i \, \X_i$ is almost surely deterministic, concentrated at a point $\alpha \in \R$. This means that the support of the distribution of the random vector ${\boldsymbol \X}$ is a subset of the hyperplane $\{\bd: \langle \ba,\bd \rangle = \alpha\}$; that is, 
$$\supp\left\{\Law({\boldsymbol \X})\right\} \subset \{\bd: \langle \ba,\bd \rangle = \alpha\}.$$
However, since ${\boldsymbol \X}$ is defined as the image of the vector-valued function 
$$\big(f(x,a,\mu_E),p(\cdot|x,a,\mu_E),e(\cdot|x)\big)$$ from $\sX \times \sA$ to $\R^m$, and as the probabilities of all image vectors are positive (since they are derived from the push-forward of the Boltzmann distribution $\nu^*_{\btheta, \blambda, \bxi}$), we must have 
$$\supp\left\{\Law({\boldsymbol \X})\right\} \not \subset \{\bd: \langle \ba,\bd \rangle = \alpha\}
$$
as $\cspan \big\{ \big(f(x,a,\mu_E),p(\cdot|x,a,\mu_E),e(\cdot|x)\big) : (x,a) \in \sX \times \sA \big\} = \R^m $ by Assumption~\ref{strong-convexity}, 
which contradicts with the above conclusion. Hence, $\Cov({\boldsymbol \X})$ is positive definite. Let $\lambda_{\min}({\boldsymbol \X})$ be the minimum eigenvalue of $\Cov({\boldsymbol \X})$, and the positive definiteness of $\Cov({\boldsymbol \X})$ ensures that $\lambda_{\min}({\boldsymbol \X}) > 0$. Since $\Cov({\boldsymbol \X})$ varies continuously with respect to $(\btheta, \blambda, \bxi)$, the minimum eigenvalue $\lambda_{\min}({\boldsymbol \X})$ also changes continuously concerning $(\btheta, \blambda, \bxi)$. This implies that if $D \subset \R^m$ is a compact subset, then
$$
\min_{(\btheta,\blambda,\bxi) \in D} \lambda_{\min}({\boldsymbol \X}) \eqqcolon  \lambda_{\min}(D) >0
$$
by uniform continuity. This means that $\Hes(g) = \Cov({\boldsymbol \X}) \succcurlyeq \lambda_{\min}(D)  \Id$ for all $(\btheta,\blambda,\bxi) \in D$, and so, $g$ is $\rho(D)$-strongly convex on $D$, where $\rho(D) \coloneqq \lambda_{\min}(D)$.  
\end{proof}

We may now introduce the gradient descent algorithm for finding the minimizer of $g$.
 
\begin{algorithm}[H]
	\caption{Gradient Descent}
	\label{GD}
	\begin{algorithmic}
		\STATE{Inputs $\left(\btheta_0,\blambda_0,\bxi_0\right), \gamma >0$}
		\STATE{Start with $(\btheta_0,\blambda_0,\bxi_0)$}
		\FOR{$k=0,\ldots,K-1$} 
        \STATE{  
			\begin{equation*}
			(\btheta_{k+1},\blambda_{k+1},\bxi_{k+1}) = (\btheta_{k},\blambda_{k},\bxi_{k}) - \gamma \, \nabla g(\btheta_{k},\blambda_{k},\bxi_{k})
			\end{equation*}
			}
	    \ENDFOR
		\RETURN{$(\btheta_{K},\blambda_{K},\bxi_{K})$ and $\nu^*_{\btheta_{K},\blambda_{K},\bxi_{K}}$}
	\end{algorithmic}
\end{algorithm}


\tk{
\begin{remark}
In gradient descent for convex functions, \( L \)-smoothness, characterized by the Lipschitz continuity of the gradient, guarantees convergence to a minimum when an appropriate learning rate \( \gamma \) is selected (i.e., \( \gamma \leq \frac{1}{L} \)). A learning rate that is too large can cause divergence or oscillations, while one that is too small can hinder convergence. The optimal learning rate typically falls within the range \( (0, \frac{1}{L}] \). Unlike adaptive learning rates, which vary based on the optimization process, a constant learning rate remains unchanged throughout this algorithm. This stability simplifies the tuning process, making implementation easier in practice.
\end{remark}
}

\begin{theorem}
Suppose that the step-size in gradient descent algorithm satisfies  $0<\gamma \leq \frac{1}{L}$. Then, we have two results of increasing strength, which depend on whether Assumption~\ref{strong-convexity} is imposed or not:
\begin{itemize}
\item [(a)] For any $k$, we have
$$
g(\btheta_{k},\blambda_{k},\bxi_{k}) - \min_{\btheta \in \R^k, \,  \blambda, \bxi \in \R^{|\sX|}} g(\btheta,\blambda,\bxi) \leq \frac{\|(\btheta_{0},\blambda_{0},\bxi_{0})-(\btheta_{*},\blambda_{*},\bxi_{*})\|^2}{2 \gamma k},
$$
where $(\btheta_{*},\blambda_{*},\bxi_{*}) \coloneqq \argmin_{\btheta \in \R^k, \,  \blambda, \bxi \in \R^{|\sX|}} g(\btheta,\blambda,\bxi)$. Moreover, 
$$
\lim_{k \rightarrow \infty} \|(\btheta_{k},\blambda_{k},\bxi_{k})-(\btheta_{*},\blambda_{*},\bxi_{*})\| = 0
$$
if we run the algorithm indefinitely. Therefore, since $\nu^*_{\btheta,\blambda,\bxi}(x,a)$ is $2M$-Lipschitz continuous with respect to $(\btheta,\blambda,\bxi)$ for all $(x,a) \in \sX \times \sA$, we also have 
$$
\lim_{k \rightarrow \infty} \|\nu^*_{\btheta_k,\blambda_k,\bxi_k}-\nu^*_{\btheta_*,\blambda_*,\bxi_*}\| = 0.
$$
\item [(b)] Suppose that Assumption~\ref{strong-convexity} holds. Define the following compact subset of $\R^m$:
$$
D \coloneqq\bigg \{\bd \in \R^m: \|\bd - (\btheta_{*},\blambda_{*},\bxi_{*})\| \leq \|(\btheta_{0},\blambda_{0},\bxi_{0})-(\btheta_{*},\blambda_{*},\bxi_{*}) \| \bigg\}.
$$
Then, for any $k$, we have
$$
\|(\btheta_{k},\blambda_{k},\bxi_{k})-(\btheta_{*},\blambda_{*},\bxi_{*})\| \leq (1-\gamma \, \rho(D))^k \, \|(\btheta_{0},\blambda_{0},\bxi_{0})-(\btheta_{*},\blambda_{*},\bxi_{*}) \|.
$$
Therefore, since $\nu^*_{\btheta,\blambda,\bxi}(x,a)$ is $2M$-Lipschitz continuous with respect to $(\btheta,\blambda,\bxi)$ for all $(x,a) \in \sX \times \sA$, we also have 
$$
\|\nu^*_{\btheta_k,\blambda_k,\bxi_k}-\nu^*_{\btheta_*,\blambda_*,\bxi_*}\| \leq \sqrt{|\sX| |\sA|} \, 2M \, (1-\gamma \, \rho(D))^k \, \|(\btheta_{0},\blambda_{0},\bxi_{0})-(\btheta_{*},\blambda_{*},\bxi_{*}) \|.
$$
\end{itemize}
\end{theorem}

\begin{proof}
Since $g$ is convex and $L$-smooth, the part (a) follows from \citet[Theorem 3.4]{GaGo23}. In the case of part (b), by examining the proof of \citet[Theorem 3.4]{GaGo23}, it becomes evident that for any $k$,
$$
\|(\btheta_{k},\blambda_{k},\bxi_{k})-(\btheta_{*},\blambda_{*},\bxi_{*})\| \leq \|(\btheta_{0},\blambda_{0},\bxi_{0})-(\btheta_{*},\blambda_{*},\bxi_{*}) \|.
$$
Hence $(\btheta_{k},\blambda_{k},\bxi_{k}) \in D$ for any $k$. Then, part (b) follows from \citet[Theorem 3.6]{GaGo23} and the fact that $g$ is $\rho(D)$-strongly convex on $D$. 
\end{proof}

Let us elaborate on the connection between the optimizer of $g$ and the policy that resolves the maximum entropy IRL problem. If $(\btheta^*,\blambda^*,\bxi^*)$ is the minimizer of $g$, which can be computed via above gradient descent algorithm, then by Proposition~\ref{a-prop1}, the optimal solution of $\mathbf{(OPT_2)}$ is 
$$
\nu^*_{\btheta^*,\blambda^*,\bxi^*}(x,a) \tk{=} \frac{e^{k_{\btheta^*,\blambda^*,\bxi^*}(x,a) }}{\sum_{(x,a) \in \sX \times \sA} e^{k_{\btheta^*,\blambda^*,\bxi^*}(x,a) }}.
$$
Then, in view of the proof of Theorem~\ref{equiv_thm}, the policy 
$$\pi_{\nu^*_{\btheta^*,\blambda^*,\bxi^*}}(a|x) = \frac{\nu^*_{\btheta^*,\blambda^*,\bxi^*}(x,a)}{\nu_{\btheta^*,\blambda^*,\bxi^*}^{*,\sX}(x)}$$
solves the maximum causal entropy problem $\mathbf{(OPT_1)}$.

\section{Mean-Field Game as \tk{a} GNEP}\label{mfg-gnep}

\ns{
}

\tk{This section diverges from the preceding one by shifting focus from IRL to forward RL in the context of MFGs. Here, we reformulate the MFG problem as a generalized Nash equilibrium problem (GNEP). We derive the MFE by leveraging established algorithms for GNEPs \citep[see][]{FaKa10,AxFaKaSa11}. This alternative solution framework complements our prior analysis of the IRL problem for MFGs, providing a comprehensive perspective on MFG solution techniques.
In the subsequent discussion, we provide a detailed rationale for our GNEP formulation, underscoring its key advantages and its broader applicability to a range of problems that may not be readily addressed by conventional MFE computation techniques \citep[see][]{LaPePeGiMuElGePi24}.}

\tk{The standard approach to compute the MFE in MFGs involves iteratively solving an optimal control problem, which reduces to finding an optimal stationary policy for an MDP parameterized by a given mean-field term $\mu \in \Pnew(\sX)$. This can be achieved using established algorithms like value iteration or $Q$-iteration. Alternatively, a policy iteration-like approach can be used, deriving an improved policy greedily from the \(Q\)-function induced by the previous policy under \(\mu\), without affecting overall convergence properties.
}

\tk{Once the optimal or improved policy \(\pi\) is determined, the invariant distribution \(\mu^{\inv}\) associated with the transition probability \(p(\,\cdot\,|x,\pi,\mu)\) can be computed. The new mean-field term is then defined as a convex combination of the current mean-field term and this invariant distribution, defining an operator \(H\) mapping \(\Pnew(\sX)\) to \(\Pnew(\sX)\) as
\[
\mu^+ = H(\mu) = (1-\alpha) \mu + \alpha \mu^{\inv}.
\]
Convergence to an MFE relies on showing \(H\) is a contraction mapping ensuring that iterates converge to a unique fixed point \(\mu^* = H(\mu^*)\). Consequently, $\mu^*$ and the corresponding policy $\pi^*$ constitute an MFE, as applied in various studies \citep[see][Section 3.4]{LaPePeGiMuElGePi24}.}

\tk{
Establishing the contraction of $H$ requires restrictive conditions on the MFG model, such as Lipschitz continuity, strong convexity, and smoothness, with additional constraints on the relevant constants. While some studies establish convergence without a contraction condition, these typically apply to continuous-time settings or specifically structured MFGs such as potential or linear-quadratic MFGs \citep[see][Section 3.4]{LaPePeGiMuElGePi24}. The GNEP approach offers a distinct advantage over standard methods by imposing more relaxed and readily verifiable conditions on the game structure. Notably, it circumvents the need for stringent conditions, such as Lipschitz continuity of system components or the contraction property of $H$, that are often required by conventional techniques.
}

\tk{
Our GNEP-based approach, while requiring twice continuous differentiability and convexity conditions on the model components, provides a new methodology for computing MFEs. While convergence relies on the invertibility of a specific Jacobian, this can be addressed practically by employing the Moore-Penrose pseudoinverse and monitoring convergence behavior. Although initialization sensitivity arises (similar to neural networks), this framework allows leveraging a wide range of existing GNEP solution algorithms \citep[see][]{FaKa10}, significantly expanding the tools available for MFE computation. 
}

\tk{Extending prior work on linear MFGs by the third author \citep[see][]{Sal23}, this paper proposes a novel approach for computing the MFE in classical non-linear MFGs. We take advantage of the flexibility of our GNEP formulation, adapting an interior-point method from \citet{AxFaKaSa11} originally designed for solving the Karush-Kuhn-Tucker (KKT) conditions of GNEPs. 
Our main approach is to formulate the MDP associated with a given mean-field term, $\mu \in \Pnew(\sX)$, as an LP problem utilizing occupation measures. This is a well-established technique in stochastic control, and we refer the reader to \citet{HeGo00, HeLa96} for an in-depth discussion of the LP formulation of MDPs. Incorporating the mean-field consistency condition directly into this LP leads to a natural representation of the MFG problem as a GNEP.}

\subsection{GNEP Formulation}\label{sub1sec3}



For a finite set $\sE$, let $\M(\sE)$ denote the set of finite signed measures on $\sE$ and $\F(\sE)$ denote the set of real functions on $\sE$ (i.e., $\M(\sE) = \F(\sE) = \R^{\sE}$). We define bilinear forms on $\bigl(\M(\sX\times\sA),\F(\sX\times\sA)\bigr)$ and on $\bigl(\M(\sX),\F(\sX)\bigr)$ as inner products
\begin{align*}
\langle \nu,v  \rangle &\coloneqq \sum_{(x,a) \in \sX\times\sA} v(x,a) \, \nu(x,a),  \\ 
\langle \mu,u  \rangle &\coloneqq \sum_{x \in \sX} u(x) \, \mu(x), 
\end{align*}
where $\nu \in \M(\sX\times\sA)$, $v \in \F(\sX\times\sA)$, $\mu \in \M(\sX)$, and $u \in \F(\sX)$.  We define the linear map $\T_{\mu}: \M(\sX\times\sA) \rightarrow \M(\sX)$ by
\begin{align}
{\T}_{\mu}\nu(\,\cdot\,) &= \nu^{\sX}(\,\cdot\,) - \beta \sum_{(x,a) \in \sX\times\sA} p_{\mu}(\,\cdot\,|x,a) \, \nu(x,a) =: \nu^{\sX} - \beta \, \nu \, p_{\mu}, \nonumber 
\end{align}
which depends on $\mu$. Recall that, for a given $\mu$, the corresponding MDP, denoted as $\text{MDP}_{\mu}$, has the following components
\begin{align*}
\left\{\sX,\sA,r_{\mu},p_{\mu},\mu\right\},
\end{align*}
where 
\begin{align*}
r_{\mu}(x,a) := r(x,a,\mu), \,\,
p_{\mu}(\,\cdot\,|\,x,a) := p(\,\cdot\,|\,x,a,\mu), \,\, x(0) \sim \mu.
\end{align*}
Then, the optimal control problem associated to $\text{MDP}_{\mu}$ is equivalent to the following equality constrained linear program \citep[see][Lemma 3.3 and Section 4]{HeGo00}:
\begin{align*}
\text{                         }&\text{maximize}_{\nu\in \M_+(\sX\times\sA)} \text{ } \langle \nu,r_{\mu} \rangle
 \\*
&\text{subject to  } {\T}_{\mu}(\nu) = (1-\beta)\mu , 
\end{align*}
where $\M_+(\sE)$ denotes the set of positive measures on the finite set $\sE$. Here, the feasible points $\nu$ of above linear program indeed corresponds to the set of (normalized) state-action occupation measures of policies. Using this LP formulation, we first establish the following result.

\begin{lemma}\label{com-lemma1}
Let $(\nu^*,\mu^*) \in \M(\sX\times\sA)_+\times\M(\sX)_+$ be a pair with the following properties:
\begin{itemize}
\item[(a)] $\nu^*$ is the optimal solution to the above LP formulation of $\text{MDP}_{\mu^*}$.
\item[(b)] $\mu^*$ satisfies the following equation
$$
\mu^*(\,\cdot\,) = \sum_{(x,a) \in \sX\times\sA} p_{\mu^*}(\,\cdot\,|x,a) \, \nu^*(x,a).
$$
\end{itemize} 
Disintegrating $\nu^*$ as $\nu^*(x,a) = \pi^*(a|x) \, \nu^{*,\sX}(x)$, the pair $(\pi^*,\mu^*)$ is an MFE for the related MFG.
\end{lemma}

\ns{
\begin{remark}
In part (b) of the \tk{above} lemma, if we had the expression
$$
\mu^*(\,\cdot\,) = \sum_{(x,a) \in \sX\times\sA} p_{\mu^*}(\,\cdot\,|x,a) \,\pi^*(a|x)  \, \mu^*(x),
$$
where $\nu^*(x,a) = \pi^*(a|x) \, \mu^*(x)$, the conclusion would be straightforward. However, in the given statement, $\pi^*(a|x) \, \mu^*(x)$ is replaced by $\nu^*(x,a)$, allowing us to avoid an additional non-linearity introduced by the product of $p_{\mu^*}(\,\cdot\,|x,a)$ and $\mu^*(\cdot)$, since both terms depend on $\mu^*$. As we will see, this expression forms a constraint for the second player in the GNEP formulation. Therefore, avoiding unnecessary non-linearities is crucial, and this is the significance of this result. 
\end{remark}
}

\begin{proof}
The proof is analogous to the linear MFG case \citep[see][Lemma 5.1]{Sal23}. For the sake of completeness, we provide a detailed proof.

Note that $\mu^*$ and $\nu^*$ are initially not assumed to be probability measures, and so, we establish this first. Since 
\begin{align}\label{first-identity}
\nu^{*,\sX} = (1-\beta) \, \mu^* + \beta \, \nu^* \, p_{\mu^*},
\end{align}
we have $\nu^*(\sX\times\sA) = (1-\beta) \, \mu^*(\sX) + \beta \, \nu^*(\sX\times\sA) \, \mu^*(\sX)$. Similarly, since 
$$
\mu^* = \nu^* \, p_{\mu^*},
$$
we have $\mu^*(\sX) = \nu^*(\sX\times\sA) \, \mu^*(\sX)$, which implies that $\nu^*$ is a probability measure. In view of this and using (\ref{first-identity}), we obtain the following:
$$
1 = (1-\beta) \, \mu^*(\sX) + \beta \, \mu^*(\sX) = \mu^*(\sX).
$$ 
Therefore, $\mu^*$ is also a probability measure. 

Note that $\nu^*$ is the optimal occupation measure of the LP formulation of 
$\text{MDP}_{\mu^*}$, and so, $\pi^*$ is the optimal policy. Hence, $\pi^* \in \Psi(\mu^*)$. Furthermore, since 
\begin{align}
\nu^{*}(x,a) &:= (1-\beta) \sum_{t=0}^{\infty} \beta^t \, E^{\pi^*,\mu^*} \left[ 1_{\{(x(t),a(t)) = (x,a)\}} \right] \nonumber \\
&= (1-\beta) \sum_{t=0}^{\infty} \beta^t \, \sPr^{\pi^*,\mu^*} \biggl[ (x(t),a(t)) = (x,a) \biggr],  \nonumber
\end{align}
we have 
\small
\begin{align*}
\nu^* \, p_{\mu^*}(\,\cdot\,) &= \sum_{(x,a) \in \sX\times\sA} p_{\mu^*}(\,\cdot\,|x,a) \, \nu^*(x,a) \\
&=  \sum_{(x,a) \in \sX\times\sA} p_{\mu^*}(\,\cdot\,|x,a) \, \bigg\{ (1-\beta) \sum_{t=0}^{\infty} \beta^t \, \sPr^{\pi^*,\mu^*} \biggl[ (x(t),a(t)) = (x,a) \biggr] \bigg\} \\
&= (1-\beta) \sum_{t=0}^{\infty} \beta^t \bigg\{ \sum_{(x,a) \in \sX\times\sA} p_{\mu^*}(\,\cdot\,|x,a) \, \sPr^{\pi^*,\mu^*} \biggl[ (x(t),a(t)) = (x,a) \biggr] \bigg\} \\
&= (1-\beta) \sum_{t=0}^{\infty} \beta^t \, \sPr^{\pi^*,\mu^*} \biggl[ x(t+1) \in \,\cdot\, \biggr] \\
&= \frac{1-\beta}{\beta} \, \sum_{t=1}^{\infty} \beta^t \, \sPr^{\pi^*,\mu^*} \biggl[ x(t) \in \,\cdot\, \biggr] + \frac{1-\beta}{\beta} \,  \sPr^{\pi^*,\mu^*} \biggl[ x(0) \in \,\cdot\, \biggr] - \frac{1-\beta}{\beta} \,  \sPr^{\pi^*,\mu^*} \biggl[ x(0) \in \,\cdot\, \biggr] \\
&= \frac{1-\beta}{\beta} \, \sum_{t=0}^{\infty} \beta^t \, \sPr^{\pi^*,\mu^*} \biggl[ x(t) \in \,\cdot\, \biggr] - \frac{1-\beta}{\beta} \, \mu^*(\,\cdot\,) \,\, \text{(as $x(0) \sim \mu^*$)} \\
&= \frac{\nu^{*,\sX}(\,\cdot\,)}{\beta} - \frac{\mu^*(\,\cdot\,)}{\beta}  + \mu^*(\,\cdot\,).
\end{align*}
\normalsize
As $\mu^* = \nu^* \, p_{\mu^*}$, the last expression implies that $\nu^{*,\sX} = \mu^*$. Hence, $\mu^*$ satisfies (see remark above)
$$
\mu^*(\,\cdot\,) = \sum_{(x,a) \in \sX\times\sA} p_{\mu^*}(\,\cdot\,|x,a) \, \pi^*(a|x) \,  \mu^*(x),
$$
that is, $\mu^* \in \Lambda(\pi^*)$. This means that $(\pi^*,\mu^*)$ is an MFE.
\end{proof}

In order to determine the MFE, it is now sufficient to compute a pair $(\nu^*,\mu^*)$, which satisfies the conditions outlined in Lemma~\ref{com-lemma1}. To achieve this, we formulate an artificial game involving two players, which naturally leads to a GNEP. Solving for the Nash equilibrium in this artificial game yields the desired pair.

In the two-player game, the first player represents the generic agent in the MFG, while the second player represents the entire population. In this setup, an additional reward function alongside $\langle \nu, r_{\mu} \rangle$ is required, where it serves as the reward for the second player. It is worth mentioning that we possess complete freedom in selecting this reward function. Hence, one can regard this extra reward as a design parameter, which can be tailored to fulfill specific objectives.

Let $h: \M(\sX \times \sA) \times \M(\sX) \to [0, \infty)$ be a continuous artificial objective function for the second player, which depends on both $(\nu, \mu)$. With this, we formulate the following GNEP.

\begin{multicols}{2}
\centering{Player 1}
\begin{align*}
\text{Given $\mu$:} \,\,
&\text{maximize}_{\nu\in \M_+(\sX\times\sA)} \text{ } \langle \nu,r_{\mu} \rangle
\nonumber \\*
&\text{subject to  } \nu^{\sX} = (1-\beta)\mu + \beta \, \nu \, p_{\mu}  
\end{align*}\\
\centering{Player 2}
\begin{align*}
\text{Given $\nu$:} \,\,
&\text{maximize}_{\mu \in \M_+(\sX)} \text{ } h(\nu,\mu)
\nonumber \\*
&\text{subject to  } \mu = \nu \, p_{\mu}
\end{align*}
\end{multicols}

It is worth noting that in this formulation, there is an interdependency between the two reward functions and the admissible strategy sets. As a result, this situation precisely fits the definition of a GNEP. This allows us to employ techniques and methods that have been developed for solving such games in our pursuit of computing the MFE. For a more comprehensive introduction to GNEPs, we refer the reader to the survey by \citet{FaKa10}.

The following result immediately follows from Lemma~\ref{com-lemma1}.

\begin{lemma}\label{com-lemma2}
If $(\nu^*,\mu^*)$ is a Nash equilibrium of the GNEP above, then $(\pi^*,\mu^*)$ is an MFE for the related MFG, where $\nu^*(x,a) = \pi^*(a|x) \, \nu^{*,\sX}(x)$.
\end{lemma}

Typically, GNEPs are expressed with inequality constraints rather than equality constraints. While it is possible to transform equality constraints into inequality constraints by duplicating them, we can instead use an alternative approach that keeps the number of constraints manageable, as shown below. Additionally, GNEPs are usually framed as a minimization problem for each agent. To align with this, we replace the terms $\langle \nu, r_{\mu} \rangle$ and $h(\nu, \mu)$ with their negative counterparts, $\langle \nu, -r_{\mu} \rangle$ (denoted as $\langle \nu, c_{\mu} \rangle$) and $-h(\nu, \mu)$ (denoted as $g(\nu, \mu)$). These changes allow us to reformulate each problem as a minimization.

\begin{multicols}{2}
\centering{Player 1}
\begin{align*}
\text{Given $\mu$:} \,\,
&\text{minimize}_{\nu\in \M_+(\sX\times\sA)} \text{ } \langle \nu,c_{\mu} \rangle
\nonumber \\*
&\text{subject to  } \nu^{\sX} \geq (1-\beta)\mu + \beta \, \nu \, p_{\mu}  
\end{align*}\\
\centering{Player 2}
\begin{align*}
\text{Given $\nu$:} \,\,
&\text{minimize}_{\mu \in \M_+(\sX)} \text{ } g(\nu,\mu)
\nonumber \\*
&\text{subject to  } \mu \geq \nu \, p_{\mu}, \,\, \langle \mu, {\bf 1} \rangle \geq 1
\end{align*}
\end{multicols}

Here, ${\bf 1}$ denotes the constant function equal to $1$. To represent the GNEP with inequality constraints, we introduce the extra constraint $\langle \mu, {\bf 1} \rangle \geq 1$, which does not substantially expand the number of constraints.

\begin{lemma}\label{com-lemma3}
If $(\nu^*,\mu^*)$ is an equilibrium solution of the GNEP with inequality constraint, then $(\pi^*,\mu^*)$ is an MFE for the related MFG, where $\nu^*(x,a) = \pi^*(a|x) \, \nu^{*,\sX}(x)$.
\end{lemma}

\ns{
\begin{proof}
We initially prove that both $\nu^*$ and $\mu^*$ are probability measures.
Since 
$$
\nu^{*,\sX} \geq (1-\beta) \, \mu^* + \beta \, \nu^* \, p_{\mu^*},
$$
we have $\nu^*(\sX\times\sA) \geq (1-\beta) \, \mu^*(\sX) + \beta \, \nu^*(\sX\times\sA) \, \mu^*(\sX)$. Similarly, since 
$$
\mu^* \geq \nu^* \, p_{\mu^*},
$$
we have $\mu^*(\sX) \geq \nu^*(\sX\times\sA) \, \mu^*(\sX)$. Hence, 
\begin{align*}
\nu^*(\sX\times\sA)  &\geq (1-\beta) \nu^*(\sX\times\sA) \, \mu^*(\sX) +  \beta \, \nu^*(\sX\times\sA) \, \mu^*(\sX) = \nu^*(\sX\times\sA) \, \mu^*(\sX). 
\end{align*}
Therefore, $\mu^*(\sX) \leq 1$ and $\nu^*(\sX\times\sA)\leq 1$. Since $\langle \mu^*, {\bf 1} \rangle = \mu^*(\sX) \geq 1$, we also have $\mu^*(\sX) = 1$, and so,
$$
\nu^*(\sX\times\sA) \geq (1-\beta) + \beta \, \nu^*(\sX\times\sA). 
$$
Hence, $\nu^*(\sX\times\sA) \geq 1$. This implies that $\nu^*(\sX\times\sA) = 1$. That is, both $\nu^*$ and $\mu^*$ are probability measures. 

In view of the above result, $(1-\beta) \, \mu^* + \beta \, \nu^* \, p_{\mu^*}$ and $\nu^* \, p_{\mu^*}$ are also probability measures. But it is known that if two probability measures $\nu$ and $\theta$ satisfy $\nu \leq \theta$ for any event, then $\nu=\theta$. Hence, 
$
\nu^{*,\sX} = (1-\beta) \, \mu^* + \beta \, \nu^* \, p_{\mu^*}$, 
$\mu^* = \nu^* \, p_{\mu^*}$
Note that given $\mu^*$, the following optimization problems are equivalent
\begin{multicols}{2}
\centering{Problem 1}
\begin{align*}
&\text{minimize}_{\nu\in \M_+(\sX\times\sA)} \text{ } \langle \nu,c_{\mu^*} \rangle
\nonumber \\*
&\text{subject to  } \nu^{\sX} \geq (1-\beta)\mu^* + \beta \, \nu \, p_{\mu^*}  
\end{align*}\\
\centering{Problem 2}
\begin{align*}
&\text{minimize}_{\nu\in \M_+(\sX\times\sA)} \text{ } \langle \nu,c_{\mu^*} \rangle
\nonumber \\*
&\text{subject to  } \nu^{\sX} = (1-\beta)\mu^* + \beta \, \nu \, p_{\mu^*}  
\end{align*}
\end{multicols}
The second problem is an LP formulation of MDP$_{\mu^*}$ and so $\nu^*$ is the optimal occupation measure. Hence, $\pi^* \in \Lambda(\mu^*)$. With the same analysis as in the proof of Lemma~\ref{com-lemma1}, we can also prove that $\nu^{*,\sX}=\mu^*$ using $\mu^* = \nu^* \, p_{\mu^*}$. Hence, by $\nu^{*,\sX} = (1-\beta) \, \mu^* + \beta \, \nu^* \, p_{\mu^*}$, we have 
$$
\mu^*(\,\cdot\,) = \sum_{(x,a) \in \sX\times\sA} p_{\mu^*}(\,\cdot\,|x,a) \, \pi^*(a|x) \,  \mu^*(x);
$$
that is, $\mu^* \in \Phi(\pi^*)$. This implies that $(\mu^*,\pi^*)$ is MFE.

\end{proof}
}

\tk{The primary novelty of our approach lies in the GNEP formulation above and its connection to the MFE. This framework enriches the MFG analysis while offering flexibility, as existing GNEP-solving algorithms \citep[see][]{FaKa10} can now be applied to compute the MFE.}


In the upcoming section, we will adapt an algorithm proposed by \citet{AxFaKaSa11} to compute an MFE. This algorithm employs an interior-point method specifically designed to solve the Karush-Kuhn-Tucker (KKT) conditions associated with the GNEP formulation.

\subsection{Computing Equilibrium of GNEP}\label{sub2sec3}

We now give a more explicit formulation of the inequality constrained GNEP that is introduced in the previous section.
\begin{multicols}{2}
\centering{Player 1}
\begin{align*}
\text{Given $\mu$:} \,\,
&\text{minimize}_{\nu\in \R^{\sX\times\sA}} \text{ } \langle \nu,c_{\mu} \rangle
\nonumber \\*
&\text{subject to  } \nu^{\sX} \geq (1-\beta)\mu + \beta \, \nu \, p_{\mu}  \\*
&\Id \cdot \, \nu \geq 0 
\end{align*}\\ 
\centering{Player 2}
\begin{align*}
\text{Given $\nu$:} \,\,
&\text{minimize}_{\mu \in \R^{\sX}} \text{ } g(\nu,\mu)
\nonumber \\*
&\text{subject to  } \mu \geq \nu \, p_{\mu}, \,\, \langle \mu, {\bf 1} \rangle \geq 1 \\*
&\Id \cdot \, \mu \geq 0 
\end{align*}
\end{multicols}

To solve this problem, we adapt an algorithm introduced by \citet{AxFaKaSa11} that employs an interior-point method leading to a solution for KKT conditions. Since we have the flexibility in selecting the auxiliary cost function $g$ for the second player, we assume that $g$ is twice-continuously differentiable and convex with respect to $\mu$ for any given $\nu$.

\ns{
\begin{assumption}
We suppose that $p_{\mu}$ and $c_{\mu}$ are twice-continuously differentiable with respect to $\mu$. Moreover, $p_{\mu}(y|x,a)$ is convex with respect to $\mu$, for all $(y,x,a) \in \sX \times \sX \times \sX$. 
\end{assumption}

In linear MFGs \citep[see][]{Sal23}, the state transition probability $p_{\mu}$ and the one-stage reward function $r_{\mu}$ are linear in $\mu$, meaning the above assumption is unnecessary for the linear case. However, in the GNEP formulation, both the objectives and constraints become nonlinear. As a result, when the policy of one agent is fixed, the optimization problem faced by the other agent turns into a nonlinear programming problem, unlike in linear MFGs.

Under these conditions, the inequality constrained GNEP meets the requirements in assumptions A1 and A2  in \citet{AxFaKaSa11}. 
}

Let us define the functions $h_1: \R^{\sX\times\sA} \times \R^{\sX} \rightarrow \R^{\sX\times\sA}\times\R^{\sX}$ and $h_2: \R^{\sX\times\sA} \times \R^{\sX} \rightarrow \R^{\sX}\times\R\times\R^{\sX}$ as 
\begin{align*}
h_1(\nu,\mu) := 
\begin{pmatrix}
-\Id \cdot \, \nu \\
-\nu^{\sX} + (1-\beta)\mu + \beta \, \nu \, p_{\mu}
\end{pmatrix}, \,\,\,\,\,\,
h_2(\nu,\mu) := 
\begin{pmatrix}
-\Id \cdot \, \mu \\
-\langle \mu, {\bf 1} \rangle + 1 \\
-\mu + \nu \, p_{\mu}
\end{pmatrix}.
\end{align*}
Then we can write the GNEP above in the following form:
\begin{multicols}{2}
\centering{Player 1}
\begin{align*}
\text{Given $\mu$:} \,\,
&\text{minimize}_{\nu\in \R^{\sX\times\sA}} \text{ } \langle \nu,c_{\mu} \rangle
\nonumber \\*
&\text{subject to  } h_1(\nu,\mu) \leq 0  
\end{align*}\\
\centering{Player 2}
\begin{align*}
\text{Given $\nu$:} \,\,
&\text{minimize}_{\mu \in \R^{\sX}} \text{ } g(\nu,\mu)
\nonumber \\*
&\text{subject to  } h_2(\nu,\mu) \leq 0
\end{align*}
\end{multicols}

Now, let us derive the joint KKT conditions for player 1 and player 2, whose solution gives a Nash equilibrium for GNEP. To this end, we start by defining the functions
\begin{align*}
L_1(\nu,\mu,\lambda) &:= \langle \nu,c_{\mu} \rangle + \langle h_1(\nu,\mu),\lambda \rangle, \\
L_2(\nu,\mu,\gamma) &:= g(\nu,\mu) + \langle h_2(\nu,\mu),\gamma \rangle,
\end{align*}
where $\lambda$ and $\gamma$ are Lagrange multipliers of player 1 and player 2, respectively. 
Setting ${\bf F}(\nu,\mu,\lambda,\gamma) := \left(\nabla_{\nu} L_1(\nu,\mu,\lambda), \nabla_{\mu} L_2(\nu,\mu,\gamma)\right)$ and ${\bf h}(\nu,\mu) := \left(h_1(\nu,\mu), h_2(\nu,\mu)\right)$, the joint KKT conditions for player 1 and player 2 can be written as 
\begin{tcolorbox} 
[colback=white!100]
\begin{center}
\vspace{-15pt}
$$
{\bf F}(\nu,\mu,\lambda,\gamma) = 0, \,\, \lambda,\gamma \geq 0, \,\, {\bf h}(\nu,\mu) \leq 0, \,\, \langle {\bf h}(\nu,\mu), (\lambda,\gamma) \rangle = 0.
$$
\end{center}
\end{tcolorbox}


\paragraph{Root Finding Problem and Algorithm}

To transform the joint KKT conditions into a root finding problem, we introduce slack variables $({\bar \lambda},{\bar \gamma})$, where ${\bar \lambda} \in \R^{\sX\times\sA}\times\R^{\sX}$ and ${\bar \gamma} \in \R^{\sX}\times\R\times\R^{\sX}$, and define 
\begin{align*}
H(z) &:= H(\nu,\mu,\lambda,\gamma,{\bar \lambda},{\bar \gamma}) :=
\begin{pmatrix}
{\bf F}(\nu,\mu,\lambda,\gamma)  \\
{\bf h}(\nu,\mu) + ({\bar \lambda},{\bar \gamma}) \\
 (\lambda,\gamma) \circ ({\bar \lambda},{\bar \gamma}) 
\end{pmatrix} \\
\end{align*}
and  $Z := \left\{ z=(\nu,\mu,\lambda,\gamma,{\bar \lambda},{\bar \gamma}): (\lambda,\gamma), ({\bar \lambda},{\bar \gamma}) \geq 0 \right\}$, where $(\lambda,\gamma) \circ ({\bar \lambda},{\bar \gamma})$ is the vector formed by diagonal elements of the outer product of the vectors $(\lambda,\gamma)$ and  $({\bar \lambda},{\bar \gamma})$. 
Then it is straightforward to show that $(\nu,\mu,\lambda,\gamma)$ satisfy joint KKT conditions if and only if $z \coloneqq (\nu,\mu,\lambda,\gamma, {\bar \lambda},{\bar \gamma})$ satisfies the constrained root finding problem 
\begin{tcolorbox} 
[colback=white!100]
\vspace{-15pt}
\begin{center}
$$H(z) = 0, \,\, z \in Z$$
\end{center}
\end{tcolorbox}
\noindent for some $({\bar \lambda},{\bar \gamma})$.

To find a solution to the constrained root finding problem, an interior-point algorithm is developed in \citet{AxFaKaSa11}. In the remainder of this section, we explain this algorithm, which depends on the potential reduction method from \citet{MoPa99}. Let $n = |\sX\times\sA| + |\sX|$ (total number of variables in GNEP) and $m := |\sX\times\sA| + 3 \, |\sX| + 1$ (total number of constraints in the GNEP). Hence $H: \R^{n} \times \R^{2m} \rightarrow \R^{n} \times \R^{2m}$ and $Z = \R^n \times \R_+^{2m}$. We take a potential function on the interior of $Z$ as 
$$
p(u,v) = K \, \log\left(\|u\|^2 + \|v\|^2 \right) - \sum_{i=1}^{2m} \log(v_i) ,
$$ 
where $K > m$, which penalizes points that are close to the boundary of $Z$ that are far from the origin. Composing $p$ and $H$, we get a potential function for the constrained root finding problem as 
$$
\psi(z) := p(H(z)),
$$
where $z \in (\intr Z) \cap H^{-1}(\intr Z) =: Z_I$.

Let $\nabla H$ denote the Jacobian of the function $H$. Before describing the algorithm, we need to impose the following condition to establish its convergence.
\begin{assumption}\label{assumption2}
For any $z \in Z_I$, the Jacobian $\nabla H(z)$ is invertible. 
\end{assumption}
In general, providing a sufficient condition for the invertibility of the Jacobian of $H$ at any point $z \in Z_I$ in terms of system components is quite challenging. Therefore, for practical applications of the below algorithm, it is preferable to substitute the inverse of $\nabla H(z)$ with its Moore-Penrose pseudo-inverse in order to be on the safe side. Let 
$$a := \begin{pmatrix}
{\bf 0}_n^T &
{\bf 1}_{2m}^T 
\end{pmatrix} \bigg/ \left\|\begin{pmatrix}
{\bf 0}_n^T &
{\bf 1}_{2m}^T 
\end{pmatrix}\right\|.$$
Then, the algorithm is given below.

\begin{algorithm}[H]
\caption{}
\label{av-H2}
\begin{algorithmic}
\STATE{Inputs: $\kappa \in (0,1)$ and $a$}
\STATE{Start with $z_0$}
\FOR{$k=0,1,2\ldots$}
\STATE{
\begin{itemize}
\item[(a)] Choose $\sigma_k \in [0,1)$ and set
$$
d_k = \left(\nabla H(z_k)\right)^{-1} \left(\sigma_k \langle a, H(z_k) \rangle \, a  - H(z_k)\right)
$$
\item[(b)] Compute a stepsize $t_k:= \max \{\kappa^l: l = 0,1,2,\ldots \}$ such that 
\begin{align}\label{aa-eq1}
z_k + t_k \, d_k \in Z_I, \,\,\, \psi(z_k + t_k \, d_k) \leq \psi(z_k) +  t_k \, \langle \nabla \psi(z_k) , d_k \rangle 
\end{align}
\item[(c)] Set $z_{k+1} := z_k + t_k \, d_k$
\end{itemize}
}
\ENDFOR
\end{algorithmic}
\end{algorithm}

Note that by \citet[Lemma 11.3.3]{FaPa03}, we have $\langle \nabla \psi(z_k) , d_k \rangle < 0$ for
\begin{align*}
d_k = \left(\nabla H(z_k)\right)^{-1} \left(\sigma_k \langle a, H(z_k) \rangle \, a  - H(z_k)\right). 
\end{align*}
Hence, one can always find $t_k$ that satisfies (\ref{aa-eq1}). The following convergence result follows from \citet[Theorems 4.3 and 4.10]{AxFaKaSa11}.

\begin{theorem}\label{com-theorem}
Under Assumption~\ref{assumption2}, pick $\sigma_k$ so that 
$$
\limsup_{k\rightarrow\infty} \sigma_k < 1.
$$
Then, the sequence $\{z_k\} := \{(\nu_k,\mu_k,\lambda_k,\gamma_k,{\bar \lambda}_k,{\bar \gamma}_k)\}$ is bounded and any accumulation point $z^* = (\nu^*,\mu^*,\lambda^*,\gamma^*,{\bar \lambda}^*,{\bar \gamma}^*)$ of $\{z_k\}$ is a solution to the constrained root finding problem $H(z^*) = 0$; that is, $H(z_k) \rightarrow 0$ as $k\rightarrow\infty$. Hence, $(\pi^*,\mu^*)$ is an MFE for the related MFG, where $\nu^*(x,a) = \pi^*(a|x) \, \nu^{*,\sX}(x)$. 
\end{theorem}

This algorithm provides a notable improvement compared to existing methods for MFGs. Through tracking the changes in the components of the vector function $H$ during each step of the algorithm, we can confidently confirm the convergence to the MFE when these components approach zero.

\section{Numerical Examples} \label{example}

In this section, we consider two numerical examples based on the malware spread model studied in \citet{SuAd19}.

\subsection{A Malware Spread Model With Two States}

We start with a malware spread model with two states, where we suppose that there are large number of agents, and each agent has a local state $x_i(t) \in \{0,1\}$, where $x_i(t) = 0$ represents the ``healthy" state and $x_i(t) = 1$ represents the ``infected" state. Each agent can take an action $a_i(t) \in \{0,1\}$, where $a_i(t) = 0$ represents ``do nothing" and $a_i(t) = 1$ represents ``repair". The dynamics is given by 
\begin{align*}
x_i(t+1) =
\begin{cases}
x_i(t) + (1-x_i(t)) \, w_i(t) & \text{if $a_i(t)=0$}, \\
0 & \text{if $a_i(t)=1$},
\end{cases}
\end{align*}
where $w_i(t) \in \{0,1\}$ is a Bernoulli random variable with success probability $q$, which gives the probability of an agent getting infected. In this setting, if an agent chooses not to take any action, they may be infected with probability $q$, but if they choose to take a repair action, they return to the healthy state. In the infinite population limit, each agent pays a cost 
$$
c\left(x,a,\mu\right) = \theta_1 \, \mu(0) \, x + (\theta_1+\theta_2) \, \mu(1) \, x + \theta_3 \, a = (\theta_1  + \theta_2 \, \mu(1)) \, x + \theta_3 \, a,
$$
where $\mu$ is the mean-field term. Here, $\theta_3$ is the cost of repair, and $(\theta_1+\theta_2 \, \mu(1))$ represents the risk of being infected. In the IRL problem, it is assumed that the variables $\theta \coloneqq (\theta_1,\theta_2,\theta_3) \in \R^3$ are unknown. Therefore, we suppose that the cost is an element of the function class
$$
\cR \coloneqq \left\{c(x,a,\mu) = \langle \theta, f(x,a,\mu) \rangle: \theta \in \R^3, \,\, f:\sX\times\sA\times\Pnew(\sX) \rightarrow \R^3 \right\},
$$
where $f(x,a,\mu) \coloneqq (x,x \cdot \mu(1),a)$.

In the first step, we consider the forward RL problem with known parameters $(\theta_1,\theta_2,\theta_3)$ and compute a corresponding mean-field equilibrium $(\pi_E,\mu_E)$ by using the GNEP formulation. Subsequently, we use this computed $(\pi_E,\mu_E)$ to generate the feature expectation vector 
$$
\langle f \rangle_{\pi_E,\mu_E} = E^{\pi_E,\mu_E}\left[\sum_{t=0}^{\infty} \beta^t \,f(x(t),a(t),\mu_E)\right]
$$
in the maximum causal entropy IRL problem to determine the policy that maximizes the causal entropy under the feature expectation constraint. 

\subsubsection{Step 1: Finding MFE}

In the infinite population limit, the stationary version of the problem is studied and the model is formulated as a GNEP, where the cost function for player 2 is taken to be the same as that of player 1. For numerical experiments, we use the following parameters $\theta_1=0.2$, $\theta_2 = 1$, $\theta_3 =0.4$, $\beta = 0.8$, $q=0.9$. We use MATLAB to perform the numerical experiments. The algorithm runs for $10000$ iterations and uses the following parameters $\sigma_k = 0.1$, $\kappa=0.001$. To perform (\ref{aa-eq1}) in Algorithm~\ref{av-H2}, we use Armijo line search.

Now let us look at the behavior of the mean-field term. Mean-field term converges to the distribution $[0.65, 0.35]$. Hence, at the equilibrium, $65\%$ of the states are healthy. Moreover, equilibrium policy converges to the conditional distributions $\pi(\,\cdot\,|0) = [0.61,0.39]$ and  $\pi(\,\cdot\,|1) = [0,1]$. Hence, once an agent is infected, then it should apply repair action with probability $1$. However, if the agent is healthy, then it should do nothing with probability $0.61$. Let $\nu_E(x,a) \coloneqq \pi_E(a|x) \, \mu_E(x)$ denote the joint distribution on $\sX \times \sA$ induced by mean-field equilibrium $(\pi_E,\mu_E)$. Then, we have 
$$
\nu_E = \begin{bmatrix}
0.3965 & 0.2535 \\0 & 0.35
\end{bmatrix}.
$$ 
Using this joint distribution we can recover both $\mu_E$ and $\pi_E$. 

%

\subsubsection{Step 2: Solving Maximum Entropy IRL Given MFE}

We now feed the MFE found in the previous section into the IRL problem  to generate the feature expectation vector and find the policy that solves the corresponding maximum causal entropy problem. We use MATLAB for the numerical computations. The gradient descent algorithm uses the following rate $\gamma = 0.5$. We stop the iteration when the each component of the gradient of $g$ becomes less than $O(10^{-2})$.

Note that the precise value of the minimum of $g$ holds less significance in this context. The key point here is that the Boltzmann distribution $\nu^*_{\btheta^*,\blambda^*,\bxi^*}$ computed at the minimizer $(\btheta^*,\blambda^*,\bxi^*)$ of $g$ is the optimal solution for $\mathbf{(OPT_2)}$. In view of the proof of Theorem~\ref{equiv_thm}, the policy 
$$\pi_{\nu^*_{\btheta^*,\blambda^*,\bxi^*}}(a|x) = \frac{\nu^*_{\btheta^*,\blambda^*,\bxi^*}(x,a)}{\nu_{\btheta^*,\blambda^*,\bxi^*}^{*,\sX}(x)}$$
solves the maximum causal entropy IRL problem.

It turns out that the gradient descent algorithm outputs the following Boltzman distribution 
$$
\nu^*_{\btheta^*,\blambda^*,\bxi^*} = \begin{bmatrix}
0.3960 & 0.2540 \\0 & 0.35
\end{bmatrix}
$$
at the minimizer $(\btheta^*,\blambda^*,\bxi^*)$ of $g$. Note that this is very close to $\nu_E$. As a result, the corresponding policy 
$$
\pi_{\nu^*_{\btheta^*,\blambda^*,\bxi^*}} = \begin{bmatrix}
0.6093 & 0.3907 \\0 & 1
\end{bmatrix}
$$
is, as expected, very close to the equilibrium policy $\pi_E$, which is unknown to the player. Recall that only the feature expectation vector is available to the player in the IRL setting. Although the equilibrium policy $\pi_E$ and the maximum causal entropy policy $\pi_{\nu^*_{\btheta^*,\blambda^*,\bxi^*}}$ might yield the same feature expectation vector under $\mu_E$, their behavior can differ significantly. In this numerical example, the resemblance between the policies $\pi_{\nu^*_{\btheta^*,\blambda^*,\bxi^*}}$ and $\pi_E$ occur due to the feature vector structure  $f(x,a,\mu) = (x,x \cdot \mu(1),a)$. Specifically, the numerical example's feature expectation matching constraint specifies that $\nu^{*,\sX}_{\btheta^*,\blambda^*,\bxi^*} = \nu_E^{\sX}$ and $\nu^{*,\sA}_{\btheta^*,\blambda^*,\bxi^*} = \nu_E^{\sA}$. With the additional constraint 
$$
\mu_E(z) = \sum_{(x,a) \in \sX \times \sA} p(z|y,a,\mu_E) \, \nu^{*,\sX}_{\btheta^*,\blambda^*,\bxi^*}(y,a) 
$$
in the maximum causal entropy problem, the equivalence of $\nu^*_{\btheta^*,\blambda^*,\bxi^*}$ and $\nu_E$ can be established. Changing the feature vector structure could potentially lead to different solutions for the maximum causal entropy problem compared to $\pi_E$, but this still leads to a mean-field equilibrium with $\mu_E$.

\ns{
\subsection{A Malware Spread Model With Ten States}

We now consider a malware spread model with ten states, i.e., \( \sX = \{0, 0.1, 0.2, \dots, 0.9\} \). The key difference between this model and the previous one lies in the transition probabilities. Specifically, the transitions are given by  
\[
p(y \mid x, a) =
\begin{cases}
\Unif(\{x, x+1, \dots, 0.9\}) & \text{if } a=0, \\
1_{\{y=0\}} & \text{if } a=1.
\end{cases}
\]
Here, choosing action \( a = 1 \) resets the system to the healthy state \( 0 \), whereas selecting \( a = 0 \) leads to a worsening state, transitioning uniformly among larger states.  

In the infinite population limit, each agent incurs a cost  
\[
c\left(x, a, \mu\right) = \theta_1 x + \theta_2 \mu_{\av} x + \theta_3 a,
\]
where \( \mu_{\av} \) represents the mean of the mean-field term. The parameter \( \theta_3 \) corresponds to the cost of repair, while \( \theta_1 + \theta_2 \mu_{\av} \) captures the infection risk.  

In the IRL problem, the parameters \( \theta \coloneqq (\theta_1, \theta_2, \theta_3) \in \mathbb{R}^3 \) are unknown. Thus, we assume that the cost function belongs to the class  
\[
\cR \coloneqq \left\{c(x, a, \mu) = \langle \theta, f(x, a, \mu) \rangle : \theta \in \mathbb{R}^3, \, f:\sX\times\sA\times\Pnew(\sX) \rightarrow \mathbb{R}^3 \right\},
\]
where the feature mapping is given by \( f(x, a, \mu) \coloneqq (x, x \mu_{\av}, a) \).

\subsubsection{Step 1: Finding MFE}

We consider the same setup as in the previous example but now with ten states. Consequently, the details of the MATLAB code remain unchanged. For numerical experiments, we use the parameter values \( \theta_1 = 0.1 \), \( \theta_2 = 1 \), \( \theta_3 = 0.4 \), and \( \beta = 0.8 \).  

Next, we examine the behavior of the mean-field term. The mean-field distribution converges to  
\[
\mu_E = [0.3338, 0.0975, 0.0429, 0.0501, 0.0601, 0.0751, 0.1001, 0.1001, 0.1001, 0.1001].
\]
Moreover, the equilibrium policy converges to the following deterministic policy:
\[
\pi_E = \begin{bmatrix}
1 & 1 & 1 & 1 & 1 & 1 & 1 & 0 & 0 & 0 \\ 
0 & 0 & 0 & 0 & 0 & 0 & 0 & 1 & 1 & 1
\end{bmatrix}.
\]  
Let \( \nu_E(x,a) \coloneqq \pi_E(a | x) \, \mu_E(x) \) denote the joint distribution on \( \sX \times \sA \) induced by the mean-field equilibrium \( (\pi_E, \mu_E) \). Then, we obtain  
\begin{multline*}
\nu_E = \\ \begin{bmatrix}
0.3338 & 0.0375 & 0.0429 & 0.0501 & 0.0601 & 0.0751 & 0.1001 & 0 & 0 & 0 \\ 
0 & 0 & 0 & 0 & 0 & 0 & 0 & 0.1001 & 0.1001 & 0.1001
\end{bmatrix}.
\end{multline*}
From this joint distribution, we can recover both \( \mu_E \) and \( \pi_E \).

\subsubsection{Step 2: Solving Maximum Entropy IRL Given MFE}

We now use the MFE found in the previous section as input to the IRL problem to compute the feature expectation vector and determine the policy that solves the corresponding maximum causal entropy problem. The numerical computations are carried out in MATLAB using a gradient descent algorithm with step size \( \gamma = 0.0025 \).  

As discussed in the previous example, the precise minimum value of \( g \) is not of primary importance. The key result is that the Boltzmann distribution \( \nu^*_{\btheta^*,\blambda^*,\bxi^*} \), computed at the minimizer \( (\btheta^*,\blambda^*,\bxi^*) \) of \( g \), serves as the optimal solution for \( \mathbf{(OPT_2)} \). By Theorem~\ref{equiv_thm}, the policy  
\[
\pi_{\nu^*_{\btheta^*,\blambda^*,\bxi^*}}(a|x) = \frac{\nu^*_{\btheta^*,\blambda^*,\bxi^*}(x,a)}{\nu_{\btheta^*,\blambda^*,\bxi^*}^{*,\sX}(x)}
\]  
solves the maximum causal entropy IRL problem.  

The gradient descent algorithm outputs the following Boltzmann distribution at the minimizer \( (\btheta^*,\blambda^*,\bxi^*) \):  
\begin{multline*}
\nu^*_{\btheta^*,\blambda^*,\bxi^*} = \\   \begin{bmatrix}
0.3308 & 0.0371 & 0.0425 & 0.0496 & 0.0591 & 0.0724 & 0.0871 & 0.0183 & 0.0024 & 0.0004 \\ 
0.0027 & 0.0003 & 0.0005 & 0.0007 & 0.0013 & 0.0031 & 0.0131 & 0.0813 & 0.0976 & 0.0998
\end{bmatrix}.
\end{multline*}
This distribution closely resembles \( \nu_E \), with only minor differences in columns 7 and 8. Consequently, the corresponding policy  
\begin{multline*}
\pi_{\nu^*_{\btheta^*,\blambda^*,\bxi^*}} = \\ \begin{bmatrix}
0.9919 & 0.9907 & 0.9890 & 0.9856 & 0.9785 & 0.9593 & 0.8696 & 0.1834 & 0.0239 & 0.0037 \\ 
0.0081 & 0.0093 & 0.0110 & 0.0144 & 0.0215 & 0.0407 & 0.1305 & 0.8166 & 0.9761 & 0.9963
\end{bmatrix}
\end{multline*}
is also close to the equilibrium policy \( \pi_E \), again differing slightly in columns 7 and 8.  

Recall that in the IRL setting, only the feature expectation vector is observable. Although \( \pi_E \) and \( \pi_{\nu^*_{\btheta^*,\blambda^*,\bxi^*}} \) may yield the same feature expectation vector under \( \mu_E \), their actual behavior can vary significantly. In this case, the near equivalence of these policies stems from the feature structure \( f(x,a,\mu) = (x, x \cdot \mu_{\av}, a) \) and the additional invariance constraint, as observed in the previous example.  
}

\section{Conclusion}\label{conc}

\tk{This paper addresses the maximum causal entropy IRL problem within the context of discrete-time MFGs under an infinite-horizon discounted-reward optimality criterion. We first formalize the maximum causal entropy IRL problem specifically for infinite-horizon MFGs, which initially presents as a non-convex optimization problem over policies. By leveraging the LP framework commonly used for MDPs, we transform this IRL problem into a convex optimization problem over state-action occupation measures. We then introduce a gradient descent algorithm, providing a guaranteed convergence rate, to compute the optimal solution.
Furthermore, we introduce a novel algorithm that recasts the MFG problem as a GNEP. This contribution is significant as it is particularly capable of computing the MFE for the forward RL problem. The practical effectiveness of our algorithm is demonstrated through numerical examples.}


\tk{This paper focuses on the theoretical analysis of the IRL problem under the simplifying assumption of linear rewards, which may not fully capture the complexity of real-world applications. The broader IRL field increasingly incorporates non-linear reward functions to better address practical scenarios. A natural extension would be to explore non-linear reward structures, potentially by leveraging the reproducing kernel Hilbert space (RKHS) framework for their representation.
}

\tk{Another further research direction involves refining the method for solving the GNEP formulation of MFGs. While we adapt an existing interior-point method, its convergence assumptions—though more relaxed than typical methods—remain somewhat restrictive. By exploiting the specific structure of the MFG-induced GNEP, future work could aim to develop less restrictive algorithms, potentially through refinements of existing methods. This presents a challenging but impactful path for advancing MFE computation.}



\acks{This work was supported by the Scientific and Technological Research Council of Turkey (TUBITAK), under Grant no: 1001-124F134.}

\appendix

\section{Maximum Causal Entropy IRL in MDPs}\label{me-irl}

In this section, we review the maximum entropy principle developed for IRL in MDPs. The literature on this topic is quite fragmented, and we believe that summarizing the existing approaches could be valuable for readers. A recent survey by \citet{GlTo22} on maximum entropy IRL for MDPs, which focuses primarily on finite-horizon problems, also provides a helpful overview. Here, we not only present the methods developed for MDPs, but we also arrange them in chronological order to illustrate the motivation behind the evolution of different types of maximum entropy principles. This structure highlights the rationale for the introduction of the maximum causal entropy principle, particularly for handling infinite-horizon problems.

A discrete-time stochastic MDP is specified by
$
\left( \sX,\sA,p,r \right),
$
where $\sX$ is a finite state space and $\sA$ is a finite action space. The components $p : \sX \times \sA  \to \sX$ and $r: \sX \times \sA \rightarrow [0,\infty)$ are the system dynamics and the one-stage reward function, respectively. Therefore, given the current state $x(t)$ and action $a(t)$, the reward $r(x(t),a(t))$ is received immediately, and the next state $x(t+1)$ evolves to a new state stochastically according to the following dynamics:
$
x(t+1) \sim p(\cdot|x(t),a(t)).
$
In this model, a randomized policy $\pi = \{\pi_t\}_{t=0}^{M}$, where $M$ is either finite or infinite, is a sequence of functions of the form $\pi_t: \sX \rightarrow \Pnew(\sA)$.

\subsection{Finite-Horizon MDPs}

In forward RL problems, the typical goal is to maximize a finite-horizon reward defined as
\begin{align}
	J(\pi,\mu_0) &= E^{\pi}\biggl[ \sum_{t=0}^{T-1} r(x(t),a(t)) \biggr], \nonumber 
\end{align}
where $T\in \mathbb{N}$ denotes the finite horizon (i.e., $M=T-1$) and $x(0) \sim \mu_0$. In this context, the one-stage reward function $r$ is known. In contrast to RL where the agent learns a policy to maximize a pre-defined reward function, IRL takes a different approach. Here, the agent is presented with a collection of trajectories generated by an expert. By analyzing these expert demonstrations, the IRL algorithm aims to deduce the reward function that the expert was implicitly trying to optimize. In essence, the expert's behavior serves as a substitute for the reward function, allowing the IRL agent to learn the underlying objective and potentially develop similar or even improved policies.

Since IRL aims to recover the reward function from expert demonstrations, the set of possible rewards needs some structure to make the problem tractable. A common assumption in IRL is that the reward can be expressed as a linear combination of feature vectors \tk{corresponding to} state and action pairs: $$
\cR \coloneqq \left\{r_{\theta}(x,a) = \langle \theta, f(x,a) \rangle: \theta \in \R^k, \,\, f:\sX\times\sA \rightarrow \R^k \right\}.
$$ 
In this setting, we suppose that an expert generates trajectories 
$
\D=\{\left(x_i(t),a_i(t)\right)_{t=0}^{T-1}\}_{i=1}^d 
$
under some optimal policy $\pi_{\opt}$. Therefore, if $d$ is large enough, by the law of large numbers, we have 
$$
\frac{1}{d} \sum_{i=1}^d \left( \sum_{t=0}^{T-1}  f(x_i(t),a_i(t)) \right) \simeq E^{\pi_{\opt}}\left[\sum_{t=0}^{T-1} f(x(t),a(t))\right] \eqqcolon \langle f \rangle_{\pi_{\opt}},
$$
where $E^{\pi_{\opt}}$ is the expectation under $\pi_{\opt}$. This suggests that the feature expectation vector, denoted by  $\langle f \rangle_{\pi_{\opt}}$ under the optimal policy $\pi_{\opt}$, is readily available.

The maximum entropy principle was introduced by \citet{ZiMa08} to address deterministic MDPs, where the state dynamics are generated by $x(t+1) = p(x(t),a(t))$ (i.e., no uncertainty in state dynamics).\footnote{\tk{Even though the system dynamics are deterministic, agents can still employ randomized policies.}} Defining the entropy of a probability distribution $P$ on a finite set $\sE$ as
$
H(P) \coloneqq -\sum_{e \in \sE} P(e) \, \log P(e),
$
the maximum entropy IRL problem in this setting can be formulated as follows:
\renewcommand\arraystretch{1.5}
\[
\begin{array}{lll}
	\mathbf{(OPT_d)} \,\, \text{maximize}_{P} \text{ } &H(P) 
	\\
	\phantom{\mathbf{(OPT_d)} \,\, } \text{subject to} & P(\tau) \geq 0 \,\, \forall \tau \in \sZ_{\spath}, \,\,\,   \sum_{\tau\in \sZ_{\spath}} P(\tau) = 1  \\
	\phantom{x} & \sum_{\tau\in \sZ_{\spath}} F(\tau) \, P(\tau)  = \langle f \rangle_{\pi_{\opt}},
\end{array}
\]
where
$
F(\tau) \coloneqq \sum_{(x,a) \in \tau} f(x,a)
$
and $\sZ_{\spath}$ is the path space defined as
$$
\sZ_{\spath} \coloneqq \left\{ \tau \in (\sX \times \sA)^T: x(t+1) = p(x(t),a(t)), \,\, t=0,\cdots,T-2 \right\}.
$$
Here, the expert behaves according to some optimal policy $\pi_{\opt}$ under some unknown reward function 
$r_{\opt}(x,a) = \langle \theta_{\opt},f(x,a) \rangle.$
However, a key challenge in standard IRL problems arises because there can be multiple values for $\theta$ (besides the optimal $\theta_{\opt}$) that can explain the observed expert trajectories. The principle of maximum entropy suggests that among all these candidates, we should favor the one with the highest entropy.
Hence, a solution of $\mathbf{(OPT_d)}$ leads to an optimal policy with minimum bias.

\tk{In stochastic MDPs, since there is independent randomness due to noise in the state dynamics, it is not possible to formulate the maximum entropy principle over the path space as before.} To see why, suppose that the problem can be formulated over the probability distributions on the path space similar to $\mathbf{(OPT_d)}$. Then, the optimal solution $P^*$ must adhere to the state dynamics, meaning that $P^*$ must be factorized in the following form:~\footnote{Here, the first part is static, and only the second part in the product can be manipulated.}
$$
P^*(\tau) = p_0(x(0)) \, \prod_{t=0}^{T-2} p(x(t+1)|x(t),a(t)) \, \prod_{t=0}^{T-1} \pi(a(t)|x(t)).
$$

\tk{However,} the optimal solution to the maximum entropy problem over the path space might not exhibit a factored form as above. This challenge can be addressed by replacing the maximum entropy principle with the maximum causal entropy principle \citep[see][]{ZiBaDe10,ZiBaDe13}. In this approach, we consider the causally conditioned probability distribution of actions given states $P(\ba||\bx) \coloneqq \prod_{t=0}^{T-1} \pi_t(a(t)|x(t))$ and maximize its entropy \[H(P(\cdot||\cdot)) \coloneqq E^{\pi} [ -\log P(\ba||\bx) ] = \sum_{t=0}^{T-1} E^{\pi} [ -\log \pi_t(a(t)|x(t))]\]  instead. Let $\C$ denote the set of causally conditioned probability distributions. Then, the maximum causal entropy IRL problem is defined as
\renewcommand\arraystretch{1.5}
\[
\begin{array}{lll}
	\mathbf{(OPT_s)} \,\, \text{maximize}_{P(\cdot||\cdot) \in \C} \text{ } & H(P(\cdot||\cdot)) 
	\\
	\phantom{\mathbf{(OPT_s)}} \,\, \text{subject to} \text{ } & \sum_{\tau \in (\sX\times\sA)^T} F(\tau) \, {\cal T}_{P(\cdot||\cdot)}(\tau)  = \langle f \rangle_{\pi_{\opt}},
\end{array}
\]
where ${\cal T}_{P(\cdot||\cdot)}(\tau) \coloneqq p_0(x(0)) \, \prod_{t=0}^{T-2} p(x(t+1)|x(t),a(t)) \, P(\ba||\bx)$. One can prove that $\mathbf{(OPT_s)}$ is convex in $P(\cdot||\cdot)$ using the fact that each constraint is linear and the objective function is concave in $P(\cdot||\cdot)$. 

To solve the maximum entropy problem $\mathbf{(OPT_s)}$, we can introduce a Lagrange multiplier $\theta \in \R^k$ to penalize the deviations from feature expectation matching constraint. The solution of the Lagrangian relaxation of $\mathbf{(OPT_s)}$ for a given $\theta$ can then be expressed by the following soft Bellman optimality equations 
$$
	Q_t^{\theta}(x,a) = r_{\theta}(x,a) + \sum_{y \in \sX} V_{t+1}^{\theta}(y) \, p(y|x,a), \,\,\,
	V_t^{\theta}(x) = \log \sum_{a \in \sA} e^{Q_t^{\theta}(x,a)} \eqqcolon \softmax_{a \in \sA} Q_t^{\theta}(x,a),
$$ 
and it can be written in the form of
$
P_{\theta}(\ba||\bx) = \prod_{t=0}^{T-1} \pi_t^{\theta}(a(t)|x(t)),
$
where $\pi^{\theta}_t(a|x) = e^{Q_t^{\theta}(x,a)-V_t^{\theta}(x)}$ \citep[see][]{ZiBaDe10,ZiBaDe13}. This implies that the optimal solution of $\mathbf{(OPT_s)}$ leads to the probability distribution on the path space of the following form 
$$
P(\tau)\propto p_0(x(0)) \, \prod_{t=0}^{T-2} p(x(t+1)|x(t),a(t)) \, e^{ Q_{t}^{\theta}(x(t),a(t))}.
$$
There are instances in the literature, for example in \citet[eq. 1]{AaSuNa20}, \citet[eq. 3]{ChZhLiWi23} or \citet[eq. 1]{FuKaSe18},  where it is asserted that the solution to the maximum causal entropy problem for stochastic MDPs is of the form 
	$$
	P(\tau) \propto p_0(x(0)) \, \prod_{t=0}^{T-2} p(x(t+1)|x(t),a(t)) \, e^{r_{\theta}(x(t),a(t))}.
	$$
However, based on the preceding calculations, it becomes evident that $r_{\theta}$ needs to be replaced with the soft Q-functions $Q_{t}^{\theta}$. To find an optimal $\theta^*$, we can either use the feature expectation matching constraint in $\mathbf{(OPT_s)}$ or we can introduce an alternative optimization problem, whose solution satisfies the feature expectation matching constraint:
\begin{align*} 
\mathbf{(\widehat{OPT_s})}  \,\, \max_{\theta \in \R^k} \sum_{\tau \in (\sX\times\sA)^T} \log P_{\theta}(\ba||\bx) \, {\cal T}_{P_{\opt}(\cdot||\cdot)}(\tau).
\end{align*}
This problem in the literature is referred to as the maximum log-likelihood estimation problem. It serves as the commonly used formulation of the maximum entropy problem within the domain of IRL when dealing with stochastic MDPs.

\subsection{Infinite-Horizon MDPs}

In the infinite-horizon scenario (i.e., $M=\infty$), we consider the discounted reward
\begin{align}
	J(\pi,\mu_0) &= E^{\pi}\biggl[ \sum_{t=0}^{\infty} \beta^t \, r(x(t),a(t)) \biggr], \nonumber 
\end{align}
where $\beta \in (0,1)$ is the discount factor and $x(0) \sim \mu_0 $. Within the field of MDPs, it is well-established that stationary Markovian policies suffice for achieving optimality with discounted rewards. Consequently, we restrict our focus to such policies, where the policy remains constant for all time steps; that is, $\pi_s=\pi_t=\pi$ for all $s,t \geq 0$.

Extending the maximum causal entropy principle to an infinite horizon poses a challenge. The difficulty arises because defining the causally conditioned probability distribution becomes ill-defined in this scenario as it involves multiplying an infinite number of terms, each being less than one. Therefore, we use the policies instead of the causally conditioned probability distribution of actions given states to formulate the problem in the infinite-horizon case \citep[see][]{ZhBlBa18}. Indeed, this is the appropriate variant of the maximum entropy principle that we use for the IRL problem in the infinite-horizon setting for MFGs. Defining the discounted causal entropy of the policy $\pi$ as
$$
H(\pi) \coloneqq \sum_{t=0}^{\infty} \beta^t E^{\pi} \left[-\log \, \pi(a(t)|x(t)) \right],
$$
the maximum discounted causal entropy IRL problem can be formulated by
\renewcommand\arraystretch{1.5}
\[
\begin{array}{lll}
	\mathbf{(OPT_{\infty})} \,\, \text{maximize}_{\pi} \text{ } &H(\pi) 
	\\
	\phantom{\mathbf{(OPT_1)}} \,\,\,\, \text{subject to} & \pi(a|x) \geq 0 \,\, \forall (x,a) \in \sX \times \sA   \\
	\phantom{x} & \sum_{a \in \sA} \pi(a|x) = 1 \,\, \forall x \in \sX \\
	\phantom{x} & \sum_{t=0}^{\infty} \beta^t \, E^{\pi}[f(x(t),a(t))] = \langle f \rangle_{\pi_{\opt}},
\end{array}
\]
where $\langle f \rangle_{\pi_{\opt}}  \coloneqq \sum_{t=0}^{\infty} \beta^t \, E^{\pi_{\opt}}[f(x(t),a(t))]$. The challenge here is the problem's lack of convexity. This arises from the last constraint, which is non-convex with respect to policies $\pi$. \citet{ZhBlBa18} convert this non-convex problem with respect to the policies into a convex one by expressing the optimization problem using state-action occupation measures. This is indeed the approach we adapt to deal with maximum causal entropy IRL problem for MFGs.

To solve $\mathbf{(OPT_{\infty})}$, a different approach can also be taken by considering the Lagrangian relaxation of the problem. This particular formulation was not explored by \citet{ZhBlBa18}. Let us introduce Lagrange multiplier $\theta \in \R^k$ to penalize the deviations from feature expectation matching constraint. Then, the solution of the Lagrangian relaxation of $\mathbf{(OPT_{\infty})}$ for a given $\theta$ can be given by the following soft Bellman optimality equations
$$
Q^{\theta}(x,a) = r_{\theta}(x,a) + \sum_{y \in \sX} V^{\theta}(y) \, p(y|x,a), \,\,\,
V^{\theta}(x) = \log \sum_{a \in \sA} e^{Q^{\theta}(x,a)},
$$
as the Lagrangian relaxation is indeed equivalent to an entropy regularized MDP with the reward function $r_{\theta}$ \citep[see][]{NeJoGo17}, and it can be shown that the solution should take the form of
$$
\pi^{\theta}(a|x) = e^{Q^{\theta}(x,a)-V^{\theta}(x)}.
$$
To find an optimal $\theta^*$, we can then either use the feature expectation matching constraint in $\mathbf{(OPT_{\infty})}$ or introduce an alternative optimization problem, whose solution satisfies the feature expectation matching constraint:
$$
\mathbf{(\widehat{OPT_{\infty}})}  \,\, \max_{\theta \in \R^k} \sum_{(x,a) \in (\sX\times\sA)} \log \pi^{\theta}(a|x) \, \gamma_{\pi_{\opt}}(x,a),
$$
where $\gamma_{\pi_{\opt}}$ is the state-action occupation measure under the expert's optimal policy $\pi_{\opt}$. Similar to the finite-horizon scenario, this problem can be conceptualized as an instance of the maximum log-likelihood estimation problem. 

\vskip 0.2in

\begin{thebibliography}{34}
\providecommand{\natexlab}[1]{#1}
\providecommand{\url}[1]{\texttt{#1}}
\expandafter\ifx\csname urlstyle\endcsname\relax
  \providecommand{\doi}[1]{doi: #1}\else
  \providecommand{\doi}{doi: \begingroup \urlstyle{rm}\Url}\fi

\bibitem[Adams et~al.(2022)Adams, Cody, and Beling]{AdCoBe22}
Stephen Adams, Tyler Cody, and Peter~A. Beling.
\newblock A survey of inverse reinforcement learning.
\newblock \emph{Artif. Intell. Rev.}, 55\penalty0 (6):\penalty0 4307--4346,
  Aug. 2022.

\bibitem[Adlakha et~al.(2015)Adlakha, Johari, and Weintraub]{AdJoWe15}
Sachin Adlakha, Ramesh Johari, and Gabriel~Y. Weintraub.
\newblock Equilibria of dynamic games with many players: Existence,
  approximation, and market structure.
\newblock \emph{Journal of Economic Theory}, 156:\penalty0 269--316, 2015.

\bibitem[Anahtarci et~al.(2023)Anahtarci, Kariksiz, and Saldi]{AnKaSa23r}
Berkay Anahtarci, Can~Deha Kariksiz, and Naci Saldi.
\newblock {Q-Learning in Regularized Mean-field Games}.
\newblock \emph{Dynamic Games and Applications}, 13\penalty0 (1):\penalty0
  89--117, March 2023.

\bibitem[Bernasconi et~al.(2023)Bernasconi, Vittori, Trov\`{o}, and
  Restelli]{BeViTrRe23}
Martino Bernasconi, E.~Vittori, F.~Trov\`{o}, and M.~Restelli.
\newblock Dealer markets: A reinforcement learning mean field game approach.
\newblock \emph{The North American Journal of Economics and Finance},
  68:\penalty0 101974, 2023.

\bibitem[Caines et~al.(2006)Caines, Huang, and Malham\'{e}]{HuMaCa06}
Peter~E. Caines, Minyi Huang, and Roland~P. Malham\'{e}.
\newblock Large population stochastic dynamic games: Closed loop
  {M}c{K}ean-{V}lasov systems and the {N}ash certainty equivalence principle.
\newblock \emph{Communications in Information Systems}, 6:\penalty0 221--252,
  2006.

\bibitem[Chen et~al.(2023{\natexlab{a}})Chen, Liu, and Di]{ChLiDi23}
Xu~Chen, Shuo Liu, and Xuan Di.
\newblock A hybrid framework of reinforcement learning and physics-informed
  deep learning for spatiotemporal mean field games.
\newblock In \emph{Proceedings of the 22nd International Conference on
  Autonomous Agents and Multiagent Systems}, AAMAS '23, pages 1079--1087,
  2023{\natexlab{a}}.

\bibitem[Chen et~al.(2022)Chen, Zhang, Liu, and Hu]{YaLiLiHu22}
Yang Chen, Libo Zhang, Jiamou Liu, and Shuyue Hu.
\newblock Individual-level inverse reinforcement learning for mean field games.
\newblock In \emph{Proceedings of the 21st International Conference on
  Autonomous Agents and Multiagent Systems}, AAMAS '22, pages 253--262, 2022.

\bibitem[Chen et~al.(2023{\natexlab{b}})Chen, Zhang, Liu, and
  Witbrock]{ChZhLiWi23}
Yang Chen, Libo Zhang, Jiamou Liu, and Michael Witbrock.
\newblock Adversarial inverse reinforcement learning for mean field games.
\newblock In \emph{Proceedings of the 22nd International Conference on
  Autonomous Agents and Multiagent Systems}, AAMAS '23, pages 1088--1096,
  2023{\natexlab{b}}.

\bibitem[Cui and Koeppl(2021)]{CuKo21}
Kai Cui and Heinz Koeppl.
\newblock Approximately solving mean field games via entropy-regularized deep
  reinforcement learning.
\newblock In \emph{Proceedings of The 24th International Conference on
  Artificial Intelligence and Statistics}, volume 130 of \emph{Proceedings of
  Machine Learning Research}, pages 1909--1917, 13--15 Apr 2021.

\bibitem[Doncel et~al.(2022)Doncel, Gast, and Gaujal]{DoGaGa22}
Josu Doncel, Nicolas Gast, and Bruno Gaujal.
\newblock A mean field game analysis for sir dynamics with vaccination.
\newblock \emph{Probability in the Engineering and Informational Sciences}, 36
  (2):\penalty0 482--499, 2022.

\bibitem[Dreves et~al.(2011)Dreves, Facchinei, Kanzow, and
  Sagratella]{AxFaKaSa11}
Axel Dreves, Francisco Facchinei, Christian Kanzow, and Simone Sagratella.
\newblock On the solution of the {KKT} conditions of generalized nash
  equilibrium problems.
\newblock \emph{SIAM Journal on Optimization}, 21\penalty0 (3):\penalty0
  1082--1108, 2011.

\bibitem[Dupuis and Ellis(1997)]{DuEl97}
Paul Dupuis and Richard~S. Ellis.
\newblock \emph{A Weak Convergence Approach to the Theory of Large Deviations}.
\newblock Wiley-Interscience, 1997.

\bibitem[Facchinei and Kanzow(2010)]{FaKa10}
Francisco Facchinei and Cristian Kanzow.
\newblock Generalized {N}ash equilibrium problems.
\newblock \emph{Annals of Operations Research}, 175:\penalty0 177--211, 2010.

\bibitem[Facchinei and Pang(2003)]{FaPa03}
Francisco Facchinei and Jong-Shi Pang.
\newblock \emph{Finite-dimensional variational inequalities and complementarity
  problems}.
\newblock Springer, 2003.

\bibitem[Fu et~al.(2018)Fu, Luo, and Levine]{FuKaSe18}
Justin Fu, Katie Luo, and Sergey Levine.
\newblock Learning robust rewards with adversarial inverse reinforcement
  learning.
\newblock In \emph{International Conference on Learning Representations}, 2018.

\bibitem[Garrigos and Gower(2023)]{GaGo23}
Guillaume Garrigos and Robert~M. Gower.
\newblock Handbook of convergence theorems for (stochastic) gradient methods.
\newblock \emph{arXiv preprint arXiv:2301.11235}, 2023.

\bibitem[Gleave and Toyer(2022)]{GlTo22}
Adam Gleave and Sam Toyer.
\newblock A primer on maximum causal entropy inverse reinforcement learning.
\newblock \emph{arXiv preprint arXiv:2203.11409}, 2022.

\bibitem[Hern\'{a}ndez-Lerma and Gonz\'{a}lez-Hern\'{a}ndez(2000)]{HeGo00}
On\'{e}simo Hern\'{a}ndez-Lerma and Juan Gonz\'{a}lez-Hern\'{a}ndez.
\newblock Constrained {M}arkov control processes in {B}orel spaces: The
  discounted case.
\newblock \emph{Math. Meth. Oper. Res.}, 52:\penalty0 271--285, 2000.

\bibitem[Hern\'{a}ndez-Lerma and Lasserre(1996)]{HeLa96}
On\'{e}simo Hern\'{a}ndez-Lerma and Jean~B. Lasserre.
\newblock \emph{Discrete-Time {M}arkov Control Processes: Basic Optimality
  Criteria}.
\newblock Springer, 1996.

\bibitem[Lasry and Lions(2007)]{LaLi07}
Jean-Michel Lasry and Pierre-Louis Lions.
\newblock Mean field games.
\newblock \emph{Japanese Journal of Mathematics}, 2:\penalty0 229--260, 2007.

\bibitem[Lauri\`{e}re et~al.(2024)Lauri\`{e}re, Perrin, Perolat, Girgin,
  Muller, Elie, Geist, and Pietquin]{LaPePeGiMuElGePi24}
Mathieu Lauri\`{e}re, Sarah Perrin, Julien Perolat, Sertan Girgin, Paul Muller,
  Romuald Elie, Matthieu Geist, and Olivier Pietquin.
\newblock Learning in mean field games: A survey.
\newblock \emph{arXiv preprint arXiv:2205.12944}, 2024.

\bibitem[Monteiro and Pang(1999)]{MoPa99}
Renato D.~C. Monteiro and Jong-Shi Pang.
\newblock A potential reduction {N}ewton method for constrained equations.
\newblock \emph{SIAM Journal on Optimization}, 9\penalty0 (3):\penalty0
  729--754, 1999.

\bibitem[Neu et~al.(2017)Neu, G\'{o}mez, and Jonsson]{NeJoGo17}
Gergely Neu, Vicen\c{c} G\'{o}mez, and Anders Jonsson.
\newblock A unified view of entropy-regularized {M}arkov decision processes.
\newblock \emph{arXiv preprint arXiv:1705.07798}, 2017.

\bibitem[Ramponi et~al.(2023)Ramponi, Kolev, Pietquin, He, Lauri\`{e}re, and
  Geist]{GiPaOlNiMaMa23}
Giorgia Ramponi, Pavel Kolev, Olivier Pietquin, Niao He, Mathieu Lauri\`{e}re,
  and Matthieu Geist.
\newblock On imitation in mean-field games.
\newblock \emph{arXiv preprint arXiv:2306.14799}, 2023.

\bibitem[Saldi(2023)]{Sal23}
Naci Saldi.
\newblock Linear mean-field games with discounted cost.
\newblock \emph{arXiv preprint arXiv:2301.06074}, 2023.

\bibitem[Sion(1958)]{Sio58}
Maurice Sion.
\newblock On general minimax theorems.
\newblock \emph{Pacific Journal of Mathematics}, 8:\penalty0 171--176, 1958.

\bibitem[Snoswell et~al.(2020)Snoswell, Singh, and Ye]{AaSuNa20}
Aaron~J. Snoswell, Surya P.~N. Singh, and Nan Ye.
\newblock Revisiting maximum entropy inverse reinforcement learning: New
  perspectives and algorithms.
\newblock In \emph{2020 IEEE Symposium Series on Computational Intelligence
  (SSCI)}, pages 241--249, 2020.

\bibitem[Subramanian and Mahajan(2019)]{SuAd19}
Jayakumar Subramanian and Aditya Mahajan.
\newblock Reinforcement learning in stationary mean-field games.
\newblock In \emph{Proceedings of the 18th International Conference on
  Autonomous Agents and Multiagent Systems}, AAMAS '19, pages 251--259, 2019.

\bibitem[Weintraub et~al.(2005)Weintraub, Benkard, and Van~Roy]{WeBeRo05}
Gabriel~Y. Weintraub, Lanier Benkard, and Benjamin Van~Roy.
\newblock Oblivious equilibrium: A mean field approximation for large-scale
  dynamic games.
\newblock In \emph{Advances in Neural Information Processing Systems}, 2005.

\bibitem[Yang et~al.(2018)Yang, Ye, Trivedi, Xu, and Zha]{YaYe18proc}
Jiachen Yang, Xiaojing Ye, Rakshit Trivedi, Huan Xu, and Hongyuan Zha.
\newblock Deep mean field games for learning optimal behavior policy of large
  populations.
\newblock In \emph{International Conference on Learning Representations}, 2018.

\bibitem[Zhou et~al.(2018)Zhou, Bloem, and Bambos]{ZhBlBa18}
Zhengyuan Zhou, Michael Bloem, and Nicholas Bambos.
\newblock Infinite time horizon maximum causal entropy inverse reinforcement
  learning.
\newblock \emph{IEEE Transactions on Automatic Control}, 63\penalty0
  (9):\penalty0 2787--2802, 2018.

\bibitem[Ziebart et~al.(2008)Ziebart, Maas, Bagnell, and Dey]{ZiMa08}
Brian~D. Ziebart, Andrew Maas, J.~Andrew Bagnell, and Anind~K. Dey.
\newblock Maximum entropy inverse reinforcement learning.
\newblock In \emph{Proceedings of the Twenty-Third AAAI Conference on
  Artificial Intelligence}, pages 1433--1438, 2008.

\bibitem[Ziebart et~al.(2010)Ziebart, Bagnell, and Dey]{ZiBaDe10}
Brian~D. Ziebart, J.~Andrew Bagnell, and Anind~K. Dey.
\newblock Modeling interaction via the principle of maximum causal entropy.
\newblock In \emph{Proceedings of the 27th International Conference on
  International Conference on Machine Learning}, ICML'10, pages 1255--1262,
  2010.

\bibitem[Ziebart et~al.(2013)Ziebart, Bagnell, and Dey]{ZiBaDe13}
Brian~D. Ziebart, J.~Andrew Bagnell, and Anind~K. Dey.
\newblock The principle of maximum causal entropy for estimating interacting
  processes.
\newblock \emph{IEEE Transactions on Information Theory}, 59\penalty0
  (4):\penalty0 1966--1980, 2013.

\end{thebibliography}

\end{document}